\documentclass[11pt]{article}
\usepackage{amsthm, amssymb, amsmath}
\usepackage{verbatim, float}
\usepackage{mathrsfs}
\usepackage{graphics}
\usepackage[usenames,dvipsnames]{pstricks}
\usepackage{epsfig}
\usepackage{pst-grad} 
\usepackage{pst-plot} 
\usepackage{cases}
\usepackage{bm}
\usepackage[toc,page]{appendix}
\usepackage{algorithm,algorithmic}
\usepackage{dsfont}
\usepackage{algorithm,algorithmic}

\usepackage[top=1.4cm, bottom=1.4cm, left=2cm, right=2.5cm]{geometry}
\usepackage{url}

\usepackage{hhline}

\usepackage{etoolbox}
\makeatletter
\patchcmd{\@makecaption}
  {\scshape}
  {}
  {}
  {}
\makeatother

\usepackage[singlelinecheck=false]{caption}
\usepackage{diagbox}

\theoremstyle{plain}
\newtheorem{theorem}{Theorem}
\newtheorem{lemma}{Lemma}

\newtheorem{corollary}{Corollary}
\newtheorem{proposition}{Proposition}
\newtheorem{claim}{Claim}
\newtheorem{construction}{Construction}

\theoremstyle{definition}
\newtheorem{definition}{Definition}
\newtheorem{example}{Example}

\newcommand{\nchoosek}[2]{\left(\begin{array}{c}#1\\#2\end{array}\right)}

\newcommand{\A}{{\mathcal A}}
\newcommand{\B}{{\mathcal B}}
\newcommand{\C}{{\mathcal C}}

\newcommand{\X}{{\mathcal X}}

\renewcommand{\S}{{\mathcal S}}

\newcommand{\fq}{\mathbb{F}_q}

\newcommand{\cR}{{\mathscr{R}}}

\newcommand{\define}{\stackrel{\mbox{\tiny $\triangle$}}{=}}

\begin{document}

\title{Set-Codes with Small Intersections and Small Discrepancies}
\author{
R. Gabrys, H. S. Dau, C. J. Colbourn and O. Milenkovic\\
R. Gabrys and O. Milenkovic are with the Coordinated Science Laboratory (CSL), \\University of Illinois at Urbana-Champaign (UIUC)\\
C. J. Colbourn is with Arizona State University\\
H. Dau is with Monash University, Melbourne, Australia \\
{\tt ryan.gabrys@gmail.com, hoang.dau@monash.edu},\\ {\tt colbourn@asu.edu, milenkov@illinois.edu}.
}
\date{}
\maketitle

\begin{abstract} We are concerned with the problem of designing large families of subsets over a common labeled ground set that have small pairwise intersections and the property that the maximum discrepancy of the label values within each of the sets is less than or equal to one. Our results, based on transversal designs, factorizations of packings and Latin rectangles, show that by jointly constructing the sets and labeling scheme, one can achieve optimal family sizes for many parameter choices. Probabilistic arguments akin to those used for pseudorandom generators lead to significantly suboptimal results when compared to the proposed combinatorial methods. The design problem considered is motivated by applications in molecular data storage and theoretical computer science.  
\end{abstract}
\section{Introduction}

In his seminal work~\cite{beck1981balanced}, Beck introduced the notion of the discrepancy of a finite family of subsets over a finite ground set as the smallest integer $d$ for which the elements
in the ground set may be labeled by $\pm 1$ so that the sum of labels in each subset is at most $d$ in absolute value. 
Set discrepancy theory, or set bicoloring theory, has since been studied and generalized in a number of different directions~\cite{doerr2003multicolour,lovasz1986discrepancy,muthukrishnan2012optimal} and it has found applications in areas as diverse as pseudorandomness and independent permutation generation~\cite{armoni1996discrepancy,saks2000low}, $\epsilon$-approximations and geometry~\cite{matouvsek1993discrepancy}, bin packing, lattice approximations and graph spectra~\cite{rothvoss2013approximating,doerr2001lattice,solymosi2009incidences}.

The goal of these, and almost all other studies of discrepancies of set families, was to establish bounds on the largest size of families of $d$-discrepancy sets for a given ground set, or to construct large set families with prescribed discrepancy values. The sets were assumed to have no special structural constraints other than those that ensure desired discrepancy properties. An exception in this context is the work of Colbourn et al.~\cite{colbourn1999bicoloring} concerning the problem of \emph{bicoloring} Steiner triple systems (STSs)~\cite{colbourn2010crc}. Steiner triple systems are set systems in which the subsets of interest satisfy additional intersection constraints, ensuring that each pair of distinct elements of the ground set appears in exactly one subset of the system. The key finding is that STSs are inherently not possible to bicolor, as detailed in~\cite{colbourn2010crc} and the follow-up work~\cite{milazzo2003strict}. 

Recently, the authors proposed a number of coding techniques for molecular storage platforms~\cite{ygm17,yazdi2015rewritable,yazdi2017portable,yazdi2018mutually,kiah2016codes,gabrys2017asymmetric,milenkovic2018exabytes}. In one such paradigm, data is recorded in terms of the locations of nicks (cuts) in naturally occurring DNA strands. In order to correct readout errors, information is encoded into sets of nicking positions that have small overlaps, i.e., into sets with small intersections. As the DNA-strand is of the form of a double helix, i.e., composed of two strands (a sense and antisense strand), nicks may be introduced on either of the two entities. To prevent disassociation of the strands due to nicking, it is desirable to distribute the nicks equally on both strands. Such a problem setting leads to a constructive set discrepancy problem, in which one desires to construct a large family of subsets of a ground set that have 
small intersection and small discrepancy $d$. In this setting, each set gives rise to two codewords in which the labels are complements of each other, as nicks on different strands can be easily distinguished from each other in the sequencing and alignment phase~\cite{nicking}.

To address the problem, we proceed in two directions. First, we examine existing (near-optimally) sized families of sets with small intersections, such as the Bose-Bush and Babai-Frankl families~\cite{bose1952orthogonal,BabaiFrankl2004} and Steiner systems~\cite{colbourn2010crc}. For the former case, we show that one can achieve the smallest possible discrepancy ($d=0$ for even-sized sets and $d=1$ for odd-sized sets) in a natural manner, by using the properties of the defining polynomials of the sets. Second, we generalize the results of~\cite{colbourn1999bicoloring} to show that no Steiner system can have optimal discrepancy. We establish upper bounds on the size of optimal discrepancy intersecting families and then proceed to describe several constructions based on \emph{packings} that have optimal discrepancy values. In the process, we invoke graph-theoretic arguments, properties of orthogonal arrays and ideas similar to those used in the Bose and Skolem constructions~\cite{Bose1939,Skolem1958}. Note that an alternative approach to address the code construction problem for fixed set sizes is to use specialized ternary constant weight codes, but we find the set discrepancy formulation easier to work with and generalize.

The paper is organized as follows. In Section~\ref{sec:problem} we precisely formulate the problem and provide simple arguments that show how near-optimal families of sets with small intersection constructed by Bose-Bush and Babai-Frankl may be balanced. Section~\ref{sec:bounds} is devoted to studying upper bounds on the size of optimal discrepancy families of intersecting sets, while Section~\ref{sec:construction} discusses various set constructions. 

\section{Problem Statement and a Sampling of Results} \label{sec:problem}

Let $[v]$ denote the set of integers $\{1,2,\ldots,v\}$. A family of subsets of $[v]$, $\mathcal{F}_v=\{{F_1,\ldots,F_s\}}$, $s\geq 2$, is \emph{$k$-regular} if for all $1 \leq j \leq s$, $|F_j|=k$. Otherwise, the family is {\em irregular}. 
The sets in $\mathcal{F}_v$ have {\em $t$-bounded intersections} if for all pairs of distinct integers $i,j \in [s],$ $|F_i \cap F_j|\leq t-1$. Naturally, when $\mathcal{F}_v$ is $k$-regular, we require that $t<k$. 

Let $L: [v] \to \{+1,-1\}$ be a labeling of the elements in $[v]$. The {\em discrepancy} of a set $F_j \in \mathcal{F}_v$ under the labeling $L$ is $D_L(F_j)=\sum_{i \in F_j} L(i)$. For fixed values of $v$ and $t$, our goal is to find the largest size $s$ of a $t$-bounded intersection family $\mathcal{F}_v$ for which there exists a labeling $L$ such that $D_L(F_j)  \in \{-1, 0, +1\}$ for all $1 \leq j \leq s$. We refer to such a set system as an \emph{extremal balanced family}.

Sets with bounded pairwise intersections have been extensively studied in the past~\cite{BabaiFrankl2004,colbourn2010crc,Stinson2004}. Well-known examples include Steiner systems, corresponding to $k$-regular families $\mathcal{F}_v$ with the property that each $t$-subset of $[v]$ belongs to exactly one member of the family~\cite{colbourn2010crc}, and $\mathcal{C}$-intersecting families of sets described in~\cite{bose1952orthogonal,BabaiFrankl2004}. In the latter, the cardinalities of the intersections of the sets are restricted to lie in a predetermined set $\mathcal{C}$. Very little is known about discrepancies of intersecting sets and extremal balanced families in particular. The problem at hand is difficult, and it appears hard to construct extremal families for arbitrary parameter values. We therefore mostly focus on some specific choices of the parameter sets.

Some families of $t$-bounded intersection sets with (near) optimal size have inherently simple labelings that ensure the balancing property. 
We identify one such family, described in~\cite{bose1952orthogonal,BabaiFrankl2004}, and in what follows describe a simple proof that this family is indeed a balanced family. 

\textbf{Construction \cite[Thm. 4.11]{BabaiFrankl2004}.}
Let $q$ be a prime power and $1 \leq t \leq k \leq q$. Set $v = kq$.  Let $\xi$ be a primitive element of the finite field $\fq$ and set $\A = \{0,1,\xi,\ldots,\xi^{k-2}\}$, so that $|\A| = k$. For each polynomial $f \in \fq[x]$, define a set of pairs of elements
\[
A_f \define \{(a,f(a)) \colon a \in \A\}.
\]
Then $|A_f| = k$. Let
\[
\C(k,q) \define \{A_f \colon f \in \fq[x], \deg(f) \leq t-1\}.
\]
Then $\C(k,q)$ is a collection of $q^t$ $k$-subsets of the set $\X \define \A \times \fq$ such that every two sets intersect in at most $t-1$ elements, because two distinct polynomials of degree $\leq t-1$ cannot intersect in more than $t-1$ points. When $v$ is large, the Babai-Frankl construction requires $q = v/k$ to be large as well.

The Ray-Chauduri-Wilson Theorem~\cite{BabaiFrankl2004} (Theorem 4.10), asserts that the size of any family $\mathcal{F}_v$ of $k$-regular sets with $k \geq t$ whose pairwise intersection cardinalities lie in some set of cardinality $t$ that satisfies $|\mathcal{F}_v| \leq \binom{v}{t}$. As an example, the set of all $t$-subsets of $[v]$ forms a $(t-1)$-intersection bounded $t$-regular family. 
This result can be strengthened when the set of allowed cardinalities equals $\{{0,1,\ldots,t-1\}}$, provided that $v > 2k^2$, for fixed $k$ and $t$. This is the best known such bound, and it is met by the Babai-Frankl construction described above. It is easy to see that the size of the family is roughly equal to $\frac{v^t}{k^t}$, which has the same growth rate with respect to $v$ as the approximate upper bound of Ray-Chauduri and Wilson $\frac{v^t}{(2k)^t}$, and dominates the other terms provided that $v$ is sufficiently large and $k,t$ are kept constant.
\begin{proposition} 
There exists a labeling $L$ of points $(a,f(a))$ of the set $\X$ so that every set
in $\C(k,q,s)$ has discrepancy equal to $0$ when $k$ is even, and discrepancy equal to $\pm 1$ when $k$ is odd. 
\end{proposition} 
\begin{proof} 
We prove the statement by constructing a suitable labeling. First, we map the elements of the set $\X$ to $[0,n-1]$ as follows. Let $r$ be a mapping such that $r(0) = 0$ and $r(a) = m+1$ if $a = \xi^m \neq 0$. Then $r(a) \in [0,k-1]$ for every $a \in \A$ and
$r(b) \in [0,q-1]$ for every $b \in \fq$. The pair $(a,b) \in \X = \A \times \fq$ is mapped to $\sigma(a,b) = qr(a) + r(b) \in [0,n-1]$. This mapping is a bijection. 

\textbf{Case 1:} $k$ is even. We claim that for every set $A_f$, half of the elements are mapped to $[0,n/2-1]$ while the other half are mapped to $[n/2,n-1]$. Indeed, for $a \in \{0,1,\xi,\ldots,\xi^{k/2-2}\} \subset \A$ and $b = f(a) \in \fq$, the pair $(a,b)$ is mapped to 
\[
\sigma(a,b) = qr(a) + r(b) \leq q(k/2-1) + (q-1) = n/2-1. 
\] 
For $a \in \{\xi^{k/2-1},\ldots,\xi^{k-2}\} \subset \A$ and $b = f(a)$, we have
\[
\sigma(a,b) = qr(a) + r(b) \geq qk/2 + 0 = n/2. 
\]
This proves the above stated property regarding the mapping $\sigma$. Based on this result, we can now associate $\pm 1$ values with the elements of $\X$ according to the following rule: assign $-1$ to $(a,b)$ if $\sigma(a,b) < n/2$, and $+1$ to $(a,b)$ if $\sigma(a,b) \geq n/2$. Then, every set in $\C(k,q,s)$ has half of the elements mapped to $-1$ and half mapped to $+1$. Equivalently, the discrepancy of every set is equal to $0$. 

\textbf{Case 2:} $k$ is odd. Similarly as for the case when $k$ is even, we can easily show that for every set $A_f$, the first $(k+1)/2$ elements are mapped to $[0,n/2+q/2-1]$, and the last $(k-1)/2$ elements are mapped to $[n/2+q/2,n-1]$. Thus, the following $\pm 1$ assignment guarantees that each set in $\C(k,q,s)$ has element sum equal to $-1$ (or by symmetry, $+1$): assign $-1$ to $(a,b)$ if $\sigma(a,b) < n/2+q/2$ and $+1$ to $(a,b)$ if $\sigma(a,b) \geq n/2 + q/2$. 

This completes the proof of the proposition. 
\end{proof} 
A few remarks are in order. First, the labeling presented is solely based on the set $\A$, and the first entry in each pair partitions the set $\X$ into $k$-groups. Hence, the balancing property is inherited from the partition of $\A$, which suggests a close connection with constructions of \emph{transversal designs.} 

A {\em transversal design} $TD(t,k,v)$\footnote{A $TD(t,k,v)$ is equivalent to a $OA(t,k,v)$ orthogonal array, which in turn is equivalent to $k-2$ mutually orthogonal Latin squares of order $n$ when $t=2$~\cite{hedayat2012orthogonal}. We refer to all these entities as transversal designs.} consists of 
\begin{enumerate}
\item A set $V$ of $kv$ elements (called points); 
\item A partition of $V$ into sets $\{ G_i : i \in [k] \},$ where each $G_i$ contains $v$ points and is called a group; 
\item A set $\B$ of $k$-subsets called blocks for which a) every block and every group intersect in exactly one point (``blocks are transverse to groups''); and b) every $t$-subset of $V$ either occurs in exactly one block or contains two or more points from a group (but not both).  
\end{enumerate}
Because no $t$-subset of elements can appear in two or more blocks, any two distinct blocks of a $TD(t,k,v)$ intersect in at most $t-1$ elements. It is well known that
\begin{itemize}
\item Whenever $q$ is a prime power and $1 \leq t \leq k \leq q+1$, there exists a $TD(t,k,q)$.
\item Whenever $q$ is a power of $2$ and $1 \leq  k \leq q+2$, there exists a $TD(3,k,q)$.
\item For any positive integers $v$ and $t$, both $TD(t,t+1,v)$ and $TD(t,t,v)$ exist.
\end{itemize}
Whenever a $TD(t,k,v)$ exists, one can assign positive labels to the points in half of the groups and negative labels to the points in the other half of the groups when $k$ is even, or nearly half when $k$ is odd. This automatically leads to a well-defined balanced family of sets, as necessarily every block meets every group in a single point. It is straightforward to see that the Babai-Frankl construction produces a transversal design, as  outlined in~\cite{stinson1981general}. However, this construction is not optimal in general, since it may be possible to add $k$-blocks to the design that intersect the groups in more than one point.

One can add additional $k$-blocks to the designs as follows. For simplicity, assume that $k$ is even and that $t \geq 3$. Pick one group of the design that lies within the set of positively labeled elements $P_{+}$ and one group of the design that lies within the set of negatively labeled elements $P_{-}$. There are $\left(\frac{k}{2}\right)^2$ such pairs of groups. By construction, any $k$-subset with $\frac{k}{2}$ points from the first group and $\frac{k}{2}$ points from the second group intersects each block of the TD in at most two points. Furthermore, each pair of such blocks intersects in at most $\lceil \frac{k}{2} \rceil$ points, so that as long as $t> \max\{{\frac{k}{2},2\}}$, the augmented TD is both balanced and satisfies the intersection constraint\footnote{We can generalize this argument to form new blocks using $s>2$ groups, half of which are labeled $+1$ and half of which are labeled $-1$. As long as $t > \max\{{\frac{k}{s}, s\}}$, the new blocks are valid provided that any two collections of $s$ groups have fewer than $\frac{t}{k/s}$ groups in common. To avoid notational clutter, we assume that $\frac{k}{s}$ is an integer.}. This observation illustrates the fact that extremal balanced families of sets with small intersections cannot be directly derived from TDs. 

Henceforth, we focus on investigating the problem of jointly constructing large intersecting families with labelings that ensure that the set discrepancies are contained in $\{{0,\pm 1\}}$. Since it is known that Steiner triple systems cannot be balanced, i.e., that for any bicoloring of a Steiner triple system there exists one ``monochromatic'' set (i.e., a set with discrepancy $+3$ or $-3$)~\cite{colbourn1999bicoloring}, we focus our attention on \emph{packings} instead~\cite{colbourn2010crc}. 

A packing $\mathcal{C}(t,k,v)$ with parameters $(t,k,v)$ is a $k$-regular family of subsets $\mathcal{F}$ of $[v]$ with the property that each $t$-element subset of $[v]$ appears in \textit{at most} one subset. This automatically ensures that any two distinct $F_i,F_j \in \mathcal{F}_v$ satisfy $|F_i \cap F_j| \leq t-1$. It is customary to refer to the subsets as blocks, and we employ both terms. In the sections to follow, we establish the existence of packings with perfect balancing properties based on explicit constructions that rely on orthogonal arrays and factorizations of graphs~\cite{bujtas2015transversal}.

The problem of determining large families of $t$-bounded intersecting sets has also been independently studied in the theoretical computer science literature, where such sets were considered for generating pseudorandom strings~\cite{trevisan2001extractors}. Most approaches use the probabilistic method. In one such setting~\cite{trevisan2001extractors}, the ground set $[v]$ is divided into $k$ disjoint intervals of size $\frac{v}{k}$. A subset $S \in [v]$ is termed ``structured'' if it contains exactly one element from each interval. A  structured set $S$ is generated by picking uniformly at random, with replacement, one element from each interval and adding it to $S$. A probabilistic argument reveals that there exists a set of $\left(\frac{vt}{k^2}\right)^t$ structured sets, each pair of which intersects in at most $t$ positions. This bound, compared to the Ray-Chauduri-Wilson bound, is smaller by a factor of $\left(\frac{t}{k^2}\right)^t$, but balancing the sets is even easier: points in half of the intervals can be labeled by $+1$ and points in the other half by $-1$ (or vice versa). If the number of intervals is even, the discrepancy of each set is $0$; if the number of intervals is odd, the discrepancy is $\pm 1$.  

\section{Upper Bounds} \label{sec:bounds}

We first derive upper bounds on the size of extremal balanced regular packings (which we refer to as balanced packings for shorthand), and then proceed to establish constructive lower bounds for some given choices of parameters. 

For ease of notation, for a given labeling $L$, let $P_{+} = \{ i \in [v] : L(i) = +1\}$, $P_{-} = \{ i \in [v] : L(i) = -1\}$, $p_{+} = |P_{+}|$, and $p_{-} = |P_{-}| = v-p_+$. 
Without loss of generality, we assume that $p_{+} \geq p_{-}$.
We use $A(t,k,v)$ to denote the size of an extremal balanced packing family with parameters $(t,k,v)$.
The following simple upper bound on $A(t,k,v)$ is based on standard counting arguments.

\begin{lemma}\label{lem:ub} 
For any labeling $L$ with label classes of size $p_{+}$ and $p_{-}$ such that $p_{+} \geq 2 \lceil \frac{t+1}{2} \rceil$, and for $t < k$, one has
$$ A(t,k,v) \leq  \frac{ \nchoosek{p_{+}}{\lfloor t/2 \rfloor}  \nchoosek{p_{-}}{\lceil t/2 \rceil}}{ \nchoosek{\lceil k/2 \rceil}{\lfloor t/2 \rfloor}  \nchoosek{\lfloor k/2 \rfloor}{\lceil t/2 \rceil}}.$$
\end{lemma}
\begin{proof} First, there are at most $\nchoosek{p_{+}}{\lfloor t/2 \rfloor}  \nchoosek{p_{-}}{\lceil t/2 \rceil}$ $t$-element subsets which contain $\lfloor t/2 \rfloor$ elements that are labeled $+1$ and $\lceil t/2 \rceil$ elements that are labeled $-1$. For every balanced set $B$ in the underlying packing, there are $\lfloor k/2 \rfloor$ and $\lceil k/2 \rceil$ of its elements that lie in the sets $P_{-}$ and $P_{+}$, respectively, or vice versa. For simplicity, let $B_{+} \subset B$ denote the set of elements that have the label $+1$, with $|B_{+}| = \lceil \frac{k}{2} \rceil$. Let $B_{-} = B \setminus B_{+},$ with $|B_{-}| = \lfloor \frac{k}{2} \rfloor$. From each such $B$, we can generate 
$$\nchoosek{|B_{+}|}{\lfloor \frac{t}{2} \rfloor}  \nchoosek{|B_{-}|}{\lceil \frac{t}{2} \rceil}=\nchoosek{\lceil k/2 \rceil}{\lfloor t/2 \rfloor}  \nchoosek{\lfloor k/2 \rfloor}{\lceil t/2 \rceil}
$$ 
subsets of size $t$. Since each $t$-element subset of $B$ appears only once, we have 
$$\nchoosek{\lceil k/2 \rceil}{\lfloor \frac{t}{2} \rfloor}  \nchoosek{\lfloor k/2 \rfloor}{\lceil \frac{t}{2} \rceil}  A(t,k,v) \leq \nchoosek{p_{+}}{\lfloor t/2 \rfloor}  \nchoosek{p_{-}}{\lceil t/2 \rceil},$$
which proves the claim.
 \end{proof}
 As may be observed from Lemma~\ref{lem:ub}, the maximum size of a regular $(t,k,v)$ packing depends only on the values of $P_{+}$ and $P_{-}$. The next simple corollary establishes which values of these parameters maximize the upper bound. The obtained bound has the same asymptotic growth as the Ray-Chaudhuri-Wilson bound, $\left(\frac{v}{k}\right)^t$, although it addresses both the intersection \emph{and} discrepancy constraints. It also shows that the construction by Bose-Bush and Frankl-Babai is near extremal in terms of satisfying both joint intersection and balancing conditions, although the construction itself was proposed for addressing intersection constraints only.
 
 \begin{corollary}\label{cor:UB} For $t < k$,
 $$ A(t,k,v) \leq \frac{\nchoosek{\lceil v/2 \rceil}{\lfloor t/2 \rfloor}  \nchoosek{\lfloor v/2 \rfloor}{\lceil t/2 \rceil}}{\nchoosek{\lceil k/2 \rceil}{\lfloor t/2 \rfloor}  \nchoosek{\lfloor k/2 \rfloor}{\lceil t/2 \rceil}} .$$
 \end{corollary}
 
In what follows, we focus our attention on constructing balanced $(t,k,v)$ packings that meet the bound from Corollary~\ref{cor:UB} with equality. Given that one can perfectly balance the Babai-Frankl sets, the natural question arises if all or some Steiner systems, in which every $t$-subset is required to appear in exactly one block, can be perfectly balanced. Theorem~\ref{th:stbalanced} states that the answer to this question is negative: perfectly balanced systems necessarily have cardinalities smaller than that of Steiner systems with the same parameters. This result complements and extends the findings of~\cite{colbourn1999bicoloring} which pertain to Steiner triple systems only and are considered in the setting of colorability of Steiner systems. 

Our proof of Theorem~\ref{th:stbalanced} follows from Lemmas~\ref{lem:even} and \ref{lem:odd} which show that the size of an balanced packing with parameters $(t,k,v)$ is strictly less than the size of a Steiner system, denoted $S(t,k,v)$, with the same parameters for the case where $t$ is even and $t$ is odd, respectively. Both proofs make use of the following claim.

\begin{claim}\label{cl:recurse} Suppose there exists a balanced Steiner system $S(t,k,v)$ where $t \geq 4$. Then there exists a balanced Steiner system $S(t-2, k-2, v-2)$.
\end{claim}
\begin{proof} Assume that $S(t,k,v)$ is balanced with label classes $P_{+}$ and $P_{-}$. A system $S(t-2,k-2,v-2)$ can be obtained by selecting two elements $e_1 \in P_{+}$ and $e_2 \in P_{-}$. Then, the balanced Steiner system $S(t-2,k-2,v-2)$ can be obtained by selecting $\Big \{ B \setminus \{ e_1, e_2 \} : B \in S(t,k,v), \{ e_1, e_2 \} \subset B \Big \}.$
More precisely, the $S(t-2,k-2,v-2)$ is formed by taking the blocks that contain both elements $e_1$ and $e_2$ and by subsequently removing the elements $e_1, e_2$ from each of the blocks.
\end{proof}

As a result of Claim~\ref{cl:recurse}, if there exists a $(t,k,v)$ balanced system $S(t,k,v)$, then there exists a $(2,k-t+2,v-t+2)$ balanced system $S(2,k-t+2,v-t+2)$ if $t$ is even. Lemma~\ref{lem:even}  uses this claim to handle the case where $t$ is even. The case where $t$ is odd requires a little more work and is handled in Lemma~\ref{lem:odd}. The next claim is used in the proofs of Lemma~\ref{lem:even} and \ref{lem:odd}. 


\begin{claim}\label{cl:cases} Suppose the Steiner system $S(2,k,v)$ is an extremal $(t,k,v)$ balanced system for odd $k$ where $v>k$. Then, 
$$ |p_{+}-p_{-}| \leq 1.$$
\end{claim}
\begin{proof} Suppose, to the contrary, that the result does not hold. Define $\hat{k} = \lceil \frac{k}{2} \rceil$. From Lemma~\ref{lem:ub}, 
$$ A(2,k,v) \leq \frac{ \frac{v^2}{4} - 1}{\hat{k}(\hat{k}-1)}.$$
On the other hand, since $S(2,k,v)$ is a Steiner system, we have
$$ |S(2,k,v)| = \frac{v(v-1)}{(2\hat{k}-1)(2\hat{k}-2)}.$$
Since $v>k$, we have $v \geq 2\hat{k}$. For a fixed $k \geq 1$, since $|S(2,k,v)|-A(2,k,v)$ increases as $v$ increases, it follows that the difference $|S(2,k,v)|-A(2,k,v)$ is minimized when $v=2\hat{k}$. However, in this case, $|S(2,k,v)|-A(2,k,v) \geq \frac{1}{k(k-1)} > 0$ and so the result follows since $k \geq t = 2$.
\end{proof}

We are now ready to handle the case where $t$ is even. 
  \begin{lemma}\label{lem:even} For even $t$ such that $t<k<v$ where $p_{+} > \frac{k+1}{2}$, and $v>2$, one has
 $$ A(t,k,v) < |S(t,k,v)|.$$
 \end{lemma}
 \begin{proof} Suppose now that the bound in the statement of Lemma~\ref{lem:even} does not hold. This implies that there exists a balanced $S(2,k-t+2,v-t+2)$ system by repeatedly applying the procedure from Claim~\ref{cl:recurse}. For shorthand, let $\hat{P}_+$ be the set of elements in $S(2,k-t+2,v-t+2)$ labeled $+1$, and let $\hat{P}_{-}$ be the set of elements labeled $-1$, and $\hat{p}_+=|\hat{P}_+|$, and $\hat{p}_{-} = |\hat{P}_{-}|$. Also, let $\hat{k}=\lceil \frac{k-t+2}{2} \rceil$. As before, we assume $P_{+}$ are the elements in $[v]$ labeled $+1$ for the system $S(t,k,v)$.

We start with the following simple claim.

\begin{claim}\label{cl:k=p} $\hat{p}_+ > \hat{k}$. \end{claim}
\begin{proof} If $S(2,k-t+2,v-t+2)$ is obtained from $S(t,k,v)$ by applying the procedure described in Claim~\ref{cl:recurse}, then $\hat{p}_+ = p_{+} - \frac{t}{2} + 1 > \frac{k+1-t+2}{2}$, which implies $\hat{p}_+ > \lceil \frac{k-t+2}{2} \rceil$ as desired.
\end{proof}

Suppose first that $k$ is even, which (since $t$ is also even) implies that each block in $S(2,k-t+2,v-t+2)$ has $\hat{k}= \frac{k-t+2}{2}$ elements labeled $+1$ and $\hat{k}$ elements labeled $-1$. Then, the bound in Lemma~\ref{lem:ub}, which appears on the RHS of the equation that follows, along with the assumption that every $(k-t+2)$-block of $S(2,k-t+2,v-t+2)$ has $\hat{k}$ elements in $\hat{P}_+$, any pair of which uniquely identifies a block requires that

\begin{align}\label{eq:state}
\frac{\nchoosek{\hat{p}_+}{2}}{\nchoosek{\hat{k}}{2}} =\frac{\hat{p}_+  \hat{p}_-}{\hat{k}^2}.
\end{align}
This, in turn, implies that 
$$\frac{\hat{p}_+-1}{\hat{k}-1} = \frac{\hat{p}_-}{\hat{k}}.$$ 
According to Claim~\ref{cl:cases}, we only need to consider the cases where $\hat{p}_+=\hat{p}_-$ or $\hat{p}_+ = \hat{p}_- + 1$. If $\hat{p}_+=\hat{p}_-$, then $\frac{\hat{p}_+-1}{\hat{k}-1} = \frac{\hat{p}_+}{\hat{k}}$, which implies that $\hat{p}_+ = \hat{k}$, and from Claim~\ref{cl:k=p} this cannot hold. If $\hat{p}_+ = \hat{p}_- + 1$, then $\frac{\hat{p}_-}{\hat{k}} = \frac{\hat{p}_-}{\hat{k}-1}$, also a contradiction.

Next, we consider the case when $k$ is odd.  Then, every $(k-t+2)$-block has either $\hat{k}$ elements from $\hat{P}_{-}$ (which implies that the block comprises $\hat{k}-1$ elements from $\hat{P}_{-}$) or $\hat{k}-1$ elements from $\hat{P}_{-}$ (and so the block has $\hat{k}$ elements from $\hat{P}_{+}$). Let $B_{+} \subset S(2,k-t+2,v-t+2)$ denote the blocks in $S(2,k-t+2,v-t+2)$ that have $\hat{k}$ elements from $\hat{P}_{+}$ and let $B_{-} \subset S(2,k-t+2,v-t+2)$ denote the blocks in $S(2,k-t+2,v-t+2)$ that have $\hat{k}$ elements from $\hat{P}_{-}$. Clearly, $|B_{+}| + |B_{-}|=\frac{(\hat{p}_+ + \hat{p}_2)(\hat{p}_+ + \hat{p}_2-1)}{(2\hat{k}-1)(2\hat{k}-2)}$ equals the total number of blocks in $S(2,k-t+2,v-t+2)$, since the total number of blocks in the Steiner system $S(2,k-t+2,v-t+2)$ is $\frac{(\hat{p}_+ + \hat{p}_2)(\hat{p}_+ + \hat{p}_2-1)}{(2\hat{k}-1)(2\hat{k}-2)}$. Note that the total number of pairs of elements in $\hat{P}_{+}$ equals the sum of the number of pairs of elements in $\hat{P}_{+} \cap B_{+}$ and the number of pairs elements in $\hat{P}_{+} \cap B_{-}$, so that
\begin{align}\label{eq:B1}
|B_{+}| \nchoosek{\hat{k}}{2} + |B_{-}| \nchoosek{\hat{k}-1}{2} = \nchoosek{\hat{p}_+}{2}.
\end{align}
Similarly, by counting the number of pairs in $\hat{P}_{-}$, we get
\begin{align}\label{eq:B2}
|B_{+}| \nchoosek{\hat{k}-1}{2} + |B_{-}| \nchoosek{\hat{k}}{2} = \nchoosek{\hat{p}_-}{2}.
\end{align}
From (\ref{eq:B1}), we have
$$|B_{+}| = \frac{\nchoosek{\hat{p}_+}{2} - |B_{-}|  \nchoosek{\hat{k}-1}{2}}{\nchoosek{\hat{k}}{2}}.$$ 
Substituting $|B_{+}|$ into (\ref{eq:B2}), we get
\begin{align*}
|B_{-}| = \frac{\nchoosek{p_{-}}{2} - \nchoosek{p_{+}}{2} \nchoosek{\hat{k}-1}{2} / \nchoosek{\hat{k}}{2}}{\nchoosek{\hat{k}}{2} - \nchoosek{\hat{k}-1}{2} \nchoosek{\hat{k}-1}{2} / \nchoosek{\hat{k}}{2}},
\end{align*}
which simplifies to 
$$|B_{-}|= \frac{\hat{p}_+(\hat{k}-2) - \hat{p}_+^2(\hat{k}-2) + \hat{p}_2 (\hat{p}_2-1)\hat{k}}{4(\hat{k}-1)^2}.$$
Next,
\begin{align*}
|B_{+}| + |B_{-}| = \frac{ \nchoosek{\hat{p}_+}{2} - |B_{-}|  \nchoosek{\hat{k}-1}{2}}{\nchoosek{\hat{k}}{2}} + |B_{-}| =   \frac{\hat{p}_+^2 + \hat{p}_2^2 - \hat{p}_+ - \hat{p}_2}{2(\hat{k}-1)^2}.
\end{align*}
Since the number of blocks in $S(2,k-t+2,v-t+2)$ equals $\frac{(\hat{p}_++\hat{p}_2)(\hat{p}_++\hat{p}_2-1)}{(2\hat{k}-1)(2\hat{k}-2)}$, we have
\begin{align*}
\frac{\hat{p}_+^2 + \hat{p}_2^2 - \hat{p}_+ - \hat{p}_2}{2(\hat{k}-1)^2} =  \frac{(\hat{p}_++\hat{p}_2)(\hat{p}_++\hat{p}_2-1)}{(2\hat{k}-1)(2\hat{k}-2)}.
\end{align*}
First, note that $\hat{k}>1$. Otherwise, if $\hat{k} = 1$, then since $\hat{k} = \frac{k+1}{2} - \frac{t}{2} + 1$, this implies that $k+1 = t$, a contradiction. Thus, we can multiply both sides of the previous equation by $\hat{k}-1$ and simplify to get that the following must hold:
$$ \frac{\hat{p}_+^2 + \hat{p}_2^2 - \hat{p}_+ - \hat{p}_2}{\hat{k}-1} =  \frac{(\hat{p}_++\hat{p}_2)(\hat{p}_++\hat{p}_2-1)}{2\hat{k}-1}. $$

Solving for $\hat{k}$, we get
\begin{align*}
\hat{k} = \frac{2 \hat{p}_+ \hat{p}_-}{2 \hat{p}_+\hat{p}_- - \hat{p}_+^2 - \hat{p}_-^2 + \hat{p}_+ + \hat{p}_-}.
\end{align*}

According to Claim~\ref{cl:cases}, we only need to consider the cases where $\hat{p}_+=\hat{p}_-$ or $\hat{p}_+ = \hat{p}_- + 1$. If $\hat{p}_+ = \hat{p}_-$, then $\hat{k} = \hat{p}_+$, which contradicts Claim~\ref{cl:k=p}. If $\hat{p}_+ = \hat{p}_- + 1$, then we also have that $\hat{k} = \hat{p}_+$, which is another contradiction. This completes the proof.
\end{proof}
The next lemma establishes the result for odd $t$. The ideas behind the proof are similar to the proof of the previous lemma and can be found in Appendix~\ref{app:lemodd}.
 
 \begin{lemma}\label{lem:odd} For odd $t$, such that $t<k<v$ where $p_{+} > \frac{k+1}{2}$ and $v > 2$, we have
 $$ A(t,k,v) < |S(t,k,v)|.$$
 \end{lemma}
  
Based on Lemmas \ref{lem:even} and \ref{lem:odd}, we have the following theorem.
 
 \begin{theorem}\label{th:stbalanced} For $t<k<v$ where $p_{+} > \frac{k+1}{2}$ and $v>2$, 
 $$ A(t,k,v) < S(t,k,v).$$
 \end{theorem}

\section{Optimal and Near-Optimal Code Constructions} \label{sec:construction}

Next we describe how to construct extremal balanced intersecting families of sets (i.e., families of sets that meet the bound in Corollary~\ref{cor:UB}) for several parameter choices. In particular, we exhibit extremal balanced set constructions based on Latin rectangles for all values of $v$ and $t=2,\, k=3$, and asymptotically extremal balanced sets for all even $v$ and $k=t+1$.

\subsection{Extremal balanced systems with parameters $(t=2,k=3,v)$}

We consider a simple construction for the parameters $t=2,k=3$ is in terms of \emph{factorizations of graphs.} 
Recall that a \emph{factor} of a graph is a subgraph with the same set of vertices as the graph. If the spanning subgraph is $r$-regular, it is an {\em $r$-factor}. 
A graph is {\em $r$-factorizable} if its edges  can be partitioned into disjoint $r$-factors. 
Then, a $1$-factor of a graph is a {\em perfect matching}, and a {\em $1$-factorization} is a partition of the graph into matchings. 
Equivalently, a $1$-factorization of a $d$-regular graph is a proper coloring of the edges with $d$ colors. 
Suppose that $K_{p_{+}}$ is a complete graph with vertex set $P_{+}$. Let $\Phi=\{ \Phi_1, \ldots, \Phi_{p_{+}-1} \}$ be a $1$-factorization of $K_{p_{+}}$. Let the triples be of the form $\{{i, a_1, a_2\}},$ where $i \in [p_{-}]$ and where the edge $(a_1, a_2) \in \Phi_i$. It is straightforward to see that the resulting system is a partial $(2, 3, v)$ system consisting of triples defined over $[v]$. When $p_{-} \neq p_{+}$, we have the following result.

\begin{lemma}\label{lem:t2k3o} Suppose that $p_{+}$ is even and that $p_{-} < p_{+}$. Then, $A(2,3,v) = \frac{p_{+} \,p_{-}}{2}$.
\end{lemma}

\begin{example} Let $P_{+} = 6$ and $P_{-} = 3$, and for simplicity, assume that $P_{+} = \{1,2,3,4,5,6\}$ and $P_{-}=\{a,b,c\}$. Then, $\Phi=\{ \Phi_1, \ldots, \Phi_{5} \}$ is a $1$-factorization of $K_{6}$ where 
\begin{align*}
\Phi_1 &= \Big \{ \{1,2\}, \{3,4\}, \{5,6\} \Big \}, 
\Phi_2 = \Big \{ \{1,4\}, \{2,6\}, \{3,5\} \Big \}, \\ 
\Phi_3 &= \Big \{ \{1,6\}, \{2,3\}, \{4,5\} \Big \},  
\Phi_4 = \Big \{ \{2,4\}, \{1,5\}, \{3,6\} \Big \}, \\ 
\Phi_5 &= \Big \{ \{1,3\}, \{2,5\}, \{4,6\} \Big \}. 
\end{align*}
According to the procedure described prior to Lemma~\ref{lem:t2k3o}, the triples are formed by adding $a$ to each set in $\Phi_1$, adding $b$ to each set in $\Phi_2$, and adding $c$ to each set in $\Phi_3$. This leads to the following triples:
\begin{align*}
&\{ a, 1, 2 \}, \{ a, 3, 4\}, \{ a, 5, 6\}, \\
&\{ b, 1, 4 \}, \{ b, 2, 6\}, \{ b, 3, 5\}, \\
&\{ c, 1, 6 \}, \{ c, 2, 3\}, \{ c, 4, 5\}.
\end{align*}
Hence, the construction outlined in Lemma~\ref{lem:t2k3o} achieves the bound $A(2,3,v) = 9$ of Lemma~\ref{lem:ub}. 
\end{example}

Next, we turn to the case where $p_{+} = p_{-}$ and $p_{+}$ is even. A simple, yet tedious argument reveals that for these parameter choices, one cannot use the factorization approach outlined in the previous section. Hence, we propose a new construction that relies on a special type of Latin rectangles~\cite{erdos1946asymptotic}; in our setting, a Latin rectangle is defined as an array of dimension $p_{+} \times p_{+}$ with entries belonging to a set of cardinality $2p_{+}$ and such that every element appears \emph{at most} once in each row and column of the array. The rows of the array are indexed by elements from $P_{+} =\{0,1,\ldots,p_{+}-1\},$ while the columns are indexed by elements from the same set, but ``boxed'' $P_{-}=\Big \{ \boxed{0}, \boxed{1}, \ldots, \boxed{p_{+} - 1} \Big \}$, so as to distinguish them from the elements in $P_{+}$. Our choice of notation is governed by the fact that we will use the values in $P_{+}$ and $P_{-}$ - unboxed or boxed - to describe indices and placements of the elements within the array. An additional requirement on the Latin rectangles is that they do not have \emph{fixed points}, i.e., elements in the array that are equal to the index of their respective row or column. To more precisely describe the fixed point constraint, let $\ell_{i,j}$ denote the element of the Latin rectangle with row index $i \in P_{+}$ and column index $j \in P_{-}$. Obviously, triples of the form $\{{i,j,\ell_{i,j}\}}$ constitute a balanced packing as long as a) there are no fixed points in the array, in which case one would have $\ell_{i,j}=i$ or $\ell_{i,j}=j$ and therefore have the same point repeated twice and b) the triples are distinct. Clearly, only half of the entries of the rectangle may be included in the packing.

The concept of Latin rectangles without fixed points is illustrated by the next example for 
which $v=16$ and $p_{+}=p_{-}=8$.

\begin{align*}
\begin{tabular}{ c | c c c c c c c c c c  }
  & \boxed{0}   & \boxed{1} & \boxed{2} & \boxed{3} & \boxed{4}  & \boxed{5} & \boxed{6}   & \boxed{7}  \\ 
 \hline
 0 & 7              & 5              & \boxed{5}  & \boxed{4} & \boxed{3}  & \boxed{2} & 3               & 1 \\  
 1 & 6              & \boxed{4} & \boxed{3}  & \boxed{2} & \boxed{1}  & 4              & 2               & 0    \\
 2 & \boxed{3} & \boxed{2} & \boxed{1}  & \boxed{0} & 5               & 3              & 1               & 7 \\
 3 & \boxed{1} & \boxed{0} & \boxed{7}  & 6              & 4               & 2              & 0               & \boxed{2} \\
 4 & \boxed{7} & \boxed{6} & 7               & 5              & 3               & 1              & \boxed{1}  & \boxed{0} \\
 5 & \boxed{5} & 0              & 6               & 4              & 2               & \boxed{0} & \boxed{7}  & \boxed{6} \\
 6 & 1              & 7              & 5               & 3              & \boxed{7}  & \boxed{6} & \boxed{5}  & \boxed{4} \\
 7 & 0              & 6              & 4               & \boxed{6} & \boxed{5}  & \boxed{4} & \boxed{3}  & 2
\end{tabular}.
\end{align*}

The above Latin rectangle leads to the following $32$ \emph{distinct} triples:
\begin{align*}
& \{0,7,\boxed{0}\}, \{1,6,\boxed{0}\}, \{0,{5},\boxed{1}\}, \{6,{7},\boxed{1}\}, \{4,7,\boxed{2}\}, \{5,6,\boxed{2}\}, \{3,{6},\boxed{3}\}, \{4,{5},\boxed{3}\}  \\
& \{2,5,\boxed{4}\}, \{3,{4},\boxed{4}\}, \{1,4,\boxed{5}\}, \{2,{3},\boxed{5}\}, \{0,{3},\boxed{6}\}, \{1,{2},\boxed{6}\}, \{0,1,\boxed{7}\}, \{2,7,\boxed{7}\},   \\
& \{0,\boxed{2},\boxed{5}\}, \{0,\boxed{3},\boxed{4}\}, \{1,\boxed{1},\boxed{4}\}, \{1,\boxed{2},\boxed{3}\},\{2,\boxed{0},\boxed{3}\}, \{2,\boxed{1},\boxed{2}\}, \{3,\boxed{0},\boxed{1}\}, \{3,\boxed{2},\boxed{7}\}, \\
& \{4,\boxed{0},\boxed{7}\}, \{4,\boxed{1},\boxed{6}\}, \{5,\boxed{0},\boxed{5}\}, \{5,\boxed{6},\boxed{7}\}, \{6,\boxed{4},\boxed{7}\}, \{6,\boxed{5},\boxed{6}\}, \{7,\boxed{3},\boxed{6}\}, \{7,\boxed{4},\boxed{5}\}.
\end{align*}
It is straightforward to check that this set of triples constitutes a balanced $(2,3,16)$ packing of maximum size. 
We describe next a general construction for Latin rectangles with the properties listed above. 


A naive approach is to create the first column of the rectangle using only elements from either $P_{+}$ or $P_{-}$, and then fill in the remaining columns using cyclic shifts of the elements in the first column. Unfortunately, such a choice necessarily leads to one column being improperly filled out. This is why we require that the first column contain elements from both $P_{+}$ and $P_{-}$. The remaining columns are obtained from the first column by applying a simple cyclic-shift operation as detailed next. To explain how the operation is performed, we note that one can decompose a rectangle into two sub-arrays: one sub-array contains entries from $P_{+}$ only, while the other contains elements from $P_{-}$ only, as shown below. Note that in the second sub-array we switched the row and column indices for ease of interpretation.
\vspace{0.1in}

\begin{tabular}{ c | c c c c c c c c c c  }
  & \boxed{0}   & \boxed{1} & \boxed{2} & \boxed{3} & \boxed{4}  & \boxed{5} & \boxed{6}   & \boxed{7}  \\ 
 \hline
 0 & 7              & 5              & &   &   &   & 3               & 1 \\  
 1 & 6              &  &   &  &    & 4              & 2               & 0    \\
 2 & & &   &   & 5               & 3              & 1               & 7 \\
 3 &  &   &    & 6              & 4               & 2              & 0               & \\
 4 &  & & 7               & 5              & 3               & 1              &    & \\
 5 & & 0              & 6               & 4              & 2               &   &    &  \\
 6 & 1              & 7              & 5               & 3              &  &   &   & \\
 7 & 0              & 6              & 4               &    &   &  &   & 2
\end{tabular}
\quad
\begin{tabular}{ c | c c c c c c c c c c  }
  & 0   & 1 & 2 & 3 & 4  & 5 & 6   & 7  \\ 
 \hline
 \boxed{0} &               &               & \boxed{3}  & \boxed{1} & \boxed{7}  & \boxed{5} &                &  \\  
 \boxed{1} &               & \boxed{4} & \boxed{2}  & \boxed{0} & \boxed{6}  &               &                &     \\
 \boxed{2} & \boxed{5} & \boxed{3} & \boxed{1}  & \boxed{7} &                &               &                &  \\
 \boxed{3} & \boxed{4} & \boxed{2} & \boxed{0}  &               &                &               &                & \boxed{6} \\
 \boxed{4} & \boxed{3} & \boxed{1} &                &               &                &               & \boxed{7}  & \boxed{5} \\
 \boxed{5} & \boxed{2} &               &                &               &                & \boxed{0} & \boxed{6}  & \boxed{4} \\
 \boxed{6} &               &               &                &               & \boxed{1}  & \boxed{7} & \boxed{5}  & \boxed{3} \\
 \boxed{7} &               &               &                & \boxed{2} & \boxed{0}  & \boxed{6} & \boxed{4}  & 
\end{tabular}\\
\vspace{0.1in}

Formally, let $\cR$ be a Latin rectangle without fixed points, and let $\cR_{P_{+}}$ be the subarray containing elements from $P_{+}$ and let $\cR_{P_{-}}$ be the transpose of the subarray containing elements from $P_{-}$. 
For simplicity, we refer to $\cR_{P_{+}}$ as a \textit{positive Latin rectangle} and $\cR_{P_{-}}$ as a \textit{negative Latin rectangle}. Clearly, $\cR$ is a Latin rectangle without fixed points if and only if $\cR_{P_{+}}$ is a positive Latin rectangle without fixed points and $\cR_{P_{-}}$ is a negative Latin rectangle without fixed points. 

With a slight abuse of notation, and for simplicity of notation, we hence omit boxes around indices and variables. Note that if $\ell_{i,j}=k$ is an entry in $\cR_{P_{+}}$, then necessarily $i \in P_{+},$ implying that $i$ is unboxed, $j \in P_{-},$ implying that $j$ is boxed. Similarly, as $k \in P_{+},$ the entry is unboxed. Similarly, if $\ell_{i,j}=k$ is an entry in $\cR_{P_{-}}$, then $i \in P_{-}$ is boxed, $j \in P_{+}$ is unboxed, and $k \in P_{-}$ is boxed. Therefore, in order to know which entries are boxed and unboxed, if suffices to know whether $\ell_{i,j}$ is an entry in a positive or negative Latin rectangle. Furthermore, all values are taken modulo $p_{+}$, regardless of the value being boxed or not.

Our procedure for constructing Latin rectangles without fixed points starts by carefully filling out the entries in the first column, i.e., the column indexed by the element $\boxed{0}$ or the element $0$. In particular, for $\ell_{i,j} \in \cR_{P_{+}}$ and $\ell_{i,j} \in \cR_{P_{-}}$, we have:
\begin{align*}
\textbf{\textit{Update rule}:}\ \text{If } \ell_{i,j} = k, \text{ then } \ell_{i-1,j+1}=k-1, 
\end{align*}
where $i,j \in \{{0,\ldots,p_{+}-1\}}$. 


In our running example, consider the entry $\ell_{1,0} = 6 \in \cR_{P_{+}}$. According to the above update rule, $\ell_{0,1} = 5$, and $\ell_{7,2}=4$. Similarly, $\ell_{2, 0} = 5 \in \cR_{P_{-}}$, which implies $\ell_{1, 1} = 4$. 
It is straightforward to see that both $\cR_{P_{+}}$ and $\cR_{P_{-}}$ are completely determined by the choice of the first column. We hence focus on describing how to properly choose this column.
%
%

The first relevant observation is that the cyclic-shift diagonal construction for Latin rectangles satisfies a symmetry condition, which asserts that the elements in the array can be grouped into pairs that lead to the same triple. As a result, the condition also describes how to choose elements in the array that lead to distinct triples in the packing, with exactly half of the elements in the array included in the packing. 

\begin{definition}\label{cl:symm} \emph{(Symmetry condition)} If $\ell_{i,j}=h \in \cR_{P_{+}},$ then
$$\ell_{h,j} = i \in \cR_{P_{+}}.$$
Similarly, if $\ell_{i,j}=h \in \cR_{P_{-}},$ then $\ell_{h,j} = i \in \cR_{P_{-}}$.
\end{definition}

A Latin rectangle without fixed points necessarily satisfies the above symmetry condition: it is straightforward to see that if an element appears at most once in each row and column of such a rectangle, any triple in the corresponding $(2,3,2p_{+})$ packing may be generated from exactly two entries in the underlying Latin rectangle. In our running example, the $(2,3,16)$ packing contains the triple $(0,7,0)$ which results from two entries in the first column, $\ell_{0,0}=7$ and $\ell_{7,0}=0$. Similarly, the triple $(0,2,5)$ arises from two entries in the first row, $\ell_{0,2}=5$ and $\ell_{0,5}= 2$. It is straightforward to check that any triple containing two elements from $P_{+}$ may be obtained from two entries in the same column, and that any triple containing two elements from $P_{-}$ may be obtained from two entries in the same column of the corresponding positive and negative rectangle, respectively. 
 
Due to the symmetry condition, we may pair up elements in the first column of $\cR_{P_{+}}$ and $\cR_{P_{-}}$ into $\frac{p_{+}}{2}$ pairs that lead to the same triples. For the running example, the first column of $\cR_{P_{+}}$ comprises two pairs $\{ \{0,7\}, \{1,6\} \}$. We refer to this set  as the \textit{positive Latin rectangle set}. Similarly, the two pairs $\{ \{2,5\}, \{3,4\} \}$ are referred to as the \textit{negative Latin rectangle set}.

Next, we introduce a \emph{regularity condition} for Latin rectangles. Since the construction of the Latin rectangles is based on cyclically shifting increasing or decreasing sequences on the diagonals, seeded by the elements in the first column, one can automatically ensure that no column contains the same element more than once. The regularity condition ensures that no \emph{row} contains the same element more than once. 

\begin{definition}\label{cl:reg} (Regularity condition) For any two distinct pairs $\{i,j\}$ and $\{i', j'\}$ in either the positive or negative Latin rectangle set, it holds that
$$ i-i' \neq j-j'.$$
\end{definition}

To intuitively justify the regularity condition, consider an example for which the regularity condition does not hold. Suppose the pairs $\{6,2\}$ and $\{4,0\}$ are in the positive Latin rectangle set. Note that since $6-4 = 2-0$, the regularity condition does not hold. For this setting, the corresponding Latin rectangle $\cR_{P_{+}}$ has $\ell_{4,0} = 0$ and $\ell_{6,0}=2$. Since $\ell_{6,0}=2$, applying our update rule twice would imply that $\ell_{4,2} = 0=\ell_{4,0}$. The last equality asserts that two entries in the row indexed by $4$ have the same value.

The following claim can be easily established by noting that the set of differences of elements in the same column of our Latin rectangles is the same for all columns.

\begin{claim}\label{cl:reg} Suppose that $\ell_{i,j} \in \cR_{P_{+}}$ or  that $\ell_{i,j} \in \cR_{P_{-}}$, and that $j \neq j'$. If $\ell_{i,j} = \ell_{i,j'}$, then there exist two distinct pairs $\{i_1, j_1\}, \{i_2, j_2\}$ in either the positive or a negative Latin rectangle set such that
$$i_2-i_1= j_2 - j_1.$$
\end{claim}
The construction procedure for Latin rectangle sets varies slightly depending on whether $p_{+}$ is divisible by four. 
We start with the simpler case when $4 |{p_{+}}$. 

\begin{construction}\label{con:con4} Suppose that $4 | p_{+}$. Set the elements of the positive Latin rectangle set to 
$$ \{i, - i - 1 \},$$
where $i \in [ \frac{p_{+}}{4} -1 ]$.  Set the elements of the negative Latin rectangle set to
$$ \{i, - i - 1 \},$$
where $i \in \{ \frac{p_{+}}{4}, \frac{p_{+}}{4} + 1, \ldots, \frac{p_{+}}{2}-1 \}$.
\end{construction}

\begin{lemma}\label{lem:4p} For $4 | p_{+}$, Construction~\ref{con:con4} results in a Latin rectangle without fixed points.   \end{lemma}
\begin{proof} 
Since the elements from $P_{+}$ within column $i$ of the underlying Latin rectangles are obtained by decreasing the values of the elements from $P_{+}$ in column $i-1$, it follows that a column contains an element from $P_{+}$ twice if and only if the first column contains an element from $P_{+}$ twice. Clearly, Construction~\ref{con:con4} ensures that the first column only contains distinct elements from $P_{+}$. A similar argument for elements from $P_{-}$. Hence, no element appears more than once in each column of the underlying Latin rectangle. 

Next, we show that no row contains an element in the rectangle more than once. According to Claim~\ref{cl:reg}, this is equivalent to showing that the regularity condition holds. Suppose, on the contrary, that there exist two distinct pairs $\{i, j\}, \{i', j' \}$ such that $i - i' = j - j'$. According to Construction~\ref{con:con4}, $j=-i-1$ and $j'=-i'-1$, so that
$$ i - i' = -(i-i'),$$
which can only hold if $i - i' = \frac{p_{+}}{2}$. However, according to Construction~\ref{con:con4}, if $i - i' = \frac{p_{+}}{2}$, then $|\{ i, i' \} \cap P_{+}| = 1$, and either $\{i,j\}$ belongs to the positive Latin rectangle set and $\{i',j'\}$ belongs to the negative Latin rectangle set or vice versa. This follows since if $\{i,j\}, \{i', j'\}$ are both in the positive Latin rectangle, then $i-i' \in \{1,2,\ldots, \frac{p_{+}}{2}-2 \}$ or $i-i' \in \{\frac{p_{+}}{2}+2,\ldots, p_{+}-2 \}$, and so $i - i' \neq \frac{p_{+}}{2}$. In either case, the condition in Claim~\ref{cl:reg} cannot be met. This implies that the regularity condition holds.

To see that the Latin rectangles do not have fixed points, suppose on the contrary that $\ell_{i,j} = i$. Then, according to our update rule, $\ell_{i,{0}} = i$, which is impossible according to Construction~\ref{con:con4}. By symmetry, the same argument can be used to handle the case where $\ell_{i,j} = j$.
\end{proof}

We consider next the case when $\frac{p_{+}}{2}$. These parameter values require a slightly different construction for the Latin rectangle sets. Once again, we illustrate the construction by an example in which $p_{+} = p_{-} = 10$, shown below.
\vspace{0.1in}

\begin{align*}
\begin{tabular}{ c | c c c c c c c c c c  }
  & \boxed{0}   & \boxed{1} & \boxed{2} & \boxed{3} & \boxed{4}  & \boxed{5} & \boxed{6}   & \boxed{7} & \boxed{8} & \boxed{9}  \\ 
 \hline
 0 & 9              & 7              & 4               & {2}            & \boxed{7}  & {8}            & 6               & \boxed{4} & 3               & 1 \\  
 1 & 8              & 5              & {3}             & \boxed{6} & 9               & 7              & \boxed{3}  & 4              & 2               & 0    \\
 2 & 6              & 4              & \boxed{5}  & 0              & 8               & \boxed{2} & 5               & 3              & 1               & 9  \\
 3 & 5              & \boxed{4} & 1               & 9              & \boxed{1}  & 6              & 4               & {2}            & 0               & 7  \\
 4 & \boxed{3} & 2              & 0               & \boxed{0} & 7               & 5              & {3}             & {1}           & 8               & 6 \\
 5 & 3              & 1              & \boxed{9}  & 8              & 6               & {4}            & {2}             & {9}           & 7               & \boxed{2} \\
 6 & 2              & \boxed{8} & 9               & 7              & {5}             & {3}            & {0}            & {8}           & \boxed{1}  & 4 \\
 7 & \boxed{7} & 0              & 8               & {6}            & {4}            & {1}            &  {9}            &\boxed{0} & 5               & 3 \\
 8 & 1              & 9              & 7               & {5}            & {2}            & {0}            & \boxed{9}  & {6}           & 4               & \boxed{6} \\
 9 & 0              & 8              & 6               & {3}            & {1}            & \boxed{8} & {7}             & 5             & \boxed{5}  & 2
\end{tabular}.
\end{align*}
\vspace{0.1in}
As before, one can decompose the Latin rectangle into two sub-arrays as shown below.
\vspace{0.1in}

\hspace{-8.0ex}\begin{tabular}{ c | c c c c c c c c c c  }
  & \boxed{0}   & \boxed{1} & \boxed{2} & \boxed{3} & \boxed{4}  & \boxed{5} & \boxed{6}   & \boxed{7} & \boxed{8} & \boxed{9}  \\ 
 \hline
 0 & 9              & 7              & 4               & {2}            &                  & {8}            & 6               &                 & 3               & 1 \\  
 1 & 8              & 5              & {3}             &                 & 9               & 7              &                  & 4              & 2               & 0    \\
 2 & 6              & 4              &                  & 0              & 8               &                 & 5               & 3              & 1               & 9  \\
 3 & 5              &                 & 1               & 9              &                  & 6              & 4               & {2}            & 0               & 7  \\
 4 &                 & 2              & 0               &                 & 7               & 5              & {3}             & {1}           & 8               & 6 \\
 5 & 3              & 1              &                  & 8              & 6               & {4}            & {2}             & {9}           & 7               &  \\
 6 & 2              &                 & 9               & 7              & {5}             & {3}            & {0}            & {8}           &                  & 4 \\
 7 &                 & 0              & 8               & {6}            & {4}            & {1}            &  {9}            &                & 5               & 3 \\
 8 & 1              & 9              & 7               & {5}            & {2}            & {0}            &                  & {6}           & 4               &    \\
 9 & 0              & 8              & 6               & {3}            & {1}            &                 & {7}             & 5             &                  & 2
\end{tabular}
\quad
\begin{tabular}{ c | c c c c c c c c c c  }
  & \boxed{0}   & \boxed{1} & \boxed{2} & \boxed{3} & \boxed{4}  & \boxed{5} & \boxed{6}   & \boxed{7} & \boxed{8} & \boxed{9}  \\ 
 \hline
 0 &                 &                 &                  &                 & \boxed{7}  &                 &                  & \boxed{4} &                  &   \\  
 1 &                 &                 &                  & \boxed{6} &                  &                 & \boxed{3}  &                 &                  &      \\
 2 &                 &                 & \boxed{5}  &                 &                  & \boxed{2} &                  &                 &                  &    \\
 3 &                 & \boxed{4} &                  &                 & \boxed{1}  &                 &                  &                 &                  &    \\
 4 & \boxed{3} &                 &                  & \boxed{0} &                  &                 &                  &                 &                  &   \\
 5 &                 &                 & \boxed{9}  &                 &                  &                 &                  &                 &                  & \boxed{2} \\
 6 &                 & \boxed{8} &                  &                 &                  &                 &                  &                 & \boxed{1}  &   \\
 7 & \boxed{7} &                 &                  &                 &                  &                 &                  &\boxed{0}  &                  &   \\
 8 &                 &                 &                  &                 &                  &                 & \boxed{9}  &                 &                  & \boxed{6} \\
 9 &                 &                 &                  &                 &                  & \boxed{8} &                  &                & \boxed{5}  &  
\end{tabular}.

\vspace{0.1in}
Recall that when $\frac{p_{+}}{2}$ is even, one partitions the first column into two sets, each of size $\frac{p_{+}}{2}$, with half of the elements lying in $P_{+}$ and half of the elements lying in $P_{-}$. For the case when $\frac{p_{+}}{2}$ is odd, adopting the same approach as described in Construction~\ref{con:con4} does not produce the desired result. We hence modify the construction of positive and negative Latin rectangle sets as follows.

\begin{construction}\label{con:conn4} Suppose that $\frac{p_{+}}{2}$ is odd. Then, the positive Latin rectangle set consists of the following pairs:
$$ \{ i, -i -1 \},$$
for $i \in [ \frac{ p_{+}-2}{4} - 1]$, and 
$$ \{ i, -i - 2 \},$$
for $i \in \{ \frac{ p_{+}-2}{4}, \frac{ p_{+}-2}{4} + 1, \ldots, \frac{ p_{+}}{2} - 2\}$. 
The negative Latin rectangle set contains only one pair:
$$ \{ \frac{p_{+}}{2} - 1, - \frac{p_{+}-2}{4}-1 \}. $$
\end{construction}

\begin{lemma}\label{lem:4np} When $\frac{p_{+}}{2}$ is odd, Construction~\ref{con:conn4} produces a Latin rectangle without fixed points.   \end{lemma}
\begin{proof} The fact that no elements appear more than once in every column and that the Latin rectangle does not have fixed points follows along the same lines as described in the proof of Lemma~\ref{lem:4p}. Hence, we focus on proving that the regularity condition is satisfied.

Suppose that two distinct pairs $\{i,j\}, \{i',j'\}$ belong to the positive Latin rectangle set and that, on the contrary, the regularity condition does not hold. Then, 
$$i' - i = j' - j.$$ 
According to Construction~\ref{con:conn4}, we need to consider four different cases:
\begin{enumerate}
\item $j = -i-1$, $j' = -i' - 1$,
\item $j = -i - 1$, $j' = -i' - 2$,
\item $j = -i - 2$, $j' = -i' - 1$,
\item $j = -i - 2$, $j' = -i' - 2$.
\end{enumerate}
Suppose first that 1.) holds. Then, $i' - i = j' - j = (-i' - 1) - (-i - 1) = i - i'$, which is true only if $i-i'=\frac{p_{+}}{2}$. However, if 1.) is true, then $i-i' \neq \frac{p_{+}}{2},$ since $i - i' \in \{ 1,2,\ldots, \frac{p_{+}}{2} -3 \}$ or  $i - i' \in \{ \frac{p_{+}}{2} + 3, \frac{p_{+}}{2} + 3, \ldots, p_{+}-2 \}$. The observation is a consequence of Construction~\ref{con:conn4}, and leads to a contradiction. Next, we consider case 2.) and suppose that $i' - i = (-i' - 2) - (-i-1) = i-i' - 1$. This clearly cannot be true as all operations are performed modulo $p_{+},$ where $p_{+}$ is even. Similarly, if 3.) holds, then $i' - i = (-i' - 1) - (-i-2) = i-i' + 1,$ which is impossible. Finally, if 4.) holds, then $i' - i = (i' - 2) - (-i - 2) = i - i'$, which is true only if $i-i'=\frac{p_{+}}{2}$. Similarly to Case 1.), if 4.) holds for $i,i'$, then $i - i'  \in  \{ 1,2,\ldots, \frac{p_{+}}{2} -3 \}$ or $i - i' \in \{ \frac{p_{+}}{2} + 3, \frac{p_{+}}{2} + 3, \ldots, p_{+}-2 \}$. Construction~\ref{con:conn4} ensures that $i-i' \neq \frac{p_{+}}{2}$, and this establishes that the regularity condition holds for all elements in $P_{+}$. To complete the proof, we need to show that the regularity condition also holds for the pair of elements in $P_{-}$. This immediately follows since 
$$\frac{p_{+}}{2} - 1 - \left( - \frac{p_{+}-2}{4}-1 \right) = \frac{3 p_{+}-2}{4} \neq  \frac{2-3 p_{+}}{4}.$$ 
\end{proof}

Lemmas~\ref{lem:4p},~\ref{lem:4np} and~\ref{lem:t2k3o} lead to the following result.

\begin{theorem}\label{th:k3ev} For even $v$, $A(2,3,v)=\lfloor \frac{v^2}{8} \rfloor$.
\end{theorem}
\begin{proof} Since $p_{+} = \frac{v}{2}$, the Latin rectangle contains $p_{+}^2$ entries, and each triple in the packing corresponds to exactly two entries in the rectangle. Hence, when $p_{+}$ is even, the resulting system contains $\frac{p_{+}^2}{2} = \frac{v^2}{8}$ triples. When $p_{+}=\frac{v}{2}$ is odd, the result follows by noting that  $\lfloor \frac{v^2}{8} \rfloor = \frac{ (p_{+}-1)\cdot (p_{+}+1)}{2}$; based on Lemma~\ref{lem:t2k3o}, we only need to assign $\frac{v}{2}-1$ elements to $P_{-}$ and $\frac{v}{2}+1$ elements to $P_{+}$. 
\end{proof}

Theorem~\ref{th:k3ev} and Lemma~\ref{lem:t2k3o} provide constructions for balanced $(2,3,v)$ packings whose cardinality matches that in the bound of Corollary~\ref{cor:UB} whenever $v \in \{4m, 4m+2, 4m+3\}$, where $m$ a positive integer. In what follows, we show how to generate an extremal balanced $(2,3,v)$ packing for $v=4m+1$ by adding one point and $p_{+}$ triples to a $(2,3,4m)$ balanced system. 

\begin{lemma} For $v=4m+1$, $A(2,3,v) = \frac{\lfloor \frac{v}{2} \rfloor \lceil \frac{v}{2} \rceil}{2}$. \end{lemma}
\begin{proof} Let $\ell_{i,j}, \,i,j \in [2m]$ be the entries of a $2m \times 2m$ Latin rectangle obtained according to Construction~1 and 2. We describe next how to append one additional column to these Latin rectangles, indexed by $\boxed{2m+1}$, and show that the resulting $2m \times (2m+1)$ Latin rectangle has no fixed points. Since the proof involves entries from the complete Latin rectangle, we use both boxed and unboxed entries for clarity of exposition. The constructed rectangle leads to a balanced $(2,3,4m+1)$ packing of maximal size using the same procedure outlined in the previous discussion.

Let $\ell'_{i,j}, $ denote the entries of the augmented $2m \times (2m+1)$ Latin rectangle, where $\ell'_{i,j} = \ell_{i,j}$ for $i \in [2m]$, $j \in \{ \boxed{1}, \ldots, \boxed{2m} \}$. We set 
$$ \ell_{i,\boxed{2m+1}} = \frac{p_{+}}{2} + i.$$
It is straightforward to verify that the resulting Latin rectangle has no fixed points since the row with entries $\{ \ell_{0,\boxed{0}}, \ell_{0,\boxed{1}}, \ldots, \ell_{0,\boxed{2m}}\}$ from the $2m \times 2m$ Latin rectangle contains only odd valued numbers in $P_{+}$. More precisely, in the $2m \times 2m$ Latin rectangles, the $i$-th row contains only odd-valued numbers from $P_{+}$ if $i$ is even, and only even-valued numbers if $i$ is odd. Therefore, for any $i$, $\ell_{i,\boxed{2m+1}} \not \in \{ \ell_{i,\boxed{0}}, \ell_{i,\boxed{1}}, \ldots, \ell_{i,\boxed{2m}}\}$. This implies that $\ell_{i,\boxed{2m+1}}$ appears once in each row. Since it is straightforward to see $\ell_{i,\boxed{2m+1}}$ appears once in column $\boxed{2m+1}$, the claim follows.
\end{proof}

Applying the described procedure on our running example results in the following Latin rectangle.
\vspace{0.1in}

\begin{align*}
\begin{tabular}{ c | c c c c c c c c c c  }
  & \boxed{0}   & \boxed{1} & \boxed{2} & \boxed{3} & \boxed{4}  & \boxed{5} & \boxed{6}   & \boxed{7} & \boxed{8}  \\ 
 \hline
 0 & 7              & 5              & \boxed{5}  & \boxed{4} & \boxed{3}  & \boxed{2} & 3               & 1               & 4 \\  
 1 & 6              & \boxed{4} & \boxed{3}  & \boxed{2} & \boxed{1}  & 4              & 2               & 0               & 5\\
 2 & \boxed{3} & \boxed{2} & \boxed{1}  & \boxed{0} & 5               & 3              & 1               & 7               & 6 \\
 3 & \boxed{1} & \boxed{0} & \boxed{7}  & 6              & 4               & 2              & 0               & \boxed{2}  & 7 \\
 4 & \boxed{7} & \boxed{6} & 7               & 5              & 3               & 1              & \boxed{1}  & \boxed{0}  & 0 \\
 5 & \boxed{5} & 0              & 6               & 4              & 2               & \boxed{0} & \boxed{7}  & \boxed{6}  & 1 \\
 6 & 1              & 7              & 5               & 3              & \boxed{7}  & \boxed{6} & \boxed{5}  & \boxed{4}  & 2 \\
 7 & 0              & 6              & 4               & \boxed{6} & \boxed{5}  & \boxed{4} & \boxed{3}  & 2               & 3
\end{tabular}.
\end{align*}

The results are summarized in the next theorem.

\begin{theorem} For any $v \geq 8$, $$A(t,k,v) = \Big \lfloor \frac{\lfloor \frac{v}{2} \rfloor \lceil \frac{v}{2} \rceil}{2} \Big \rfloor .$$
\end{theorem}

\subsection{Balanced systems with parameters $(t,k=t+1,v)$}

In the previous section, we established that for the special case $t=2$ and $k=t+1$, extremal balanced families may be constructed using specialized Latin rectangles. In what follows, we demonstrate using a different proof technique that \emph{asymptotically} extremal balanced families may be obtained for all choices of $t$ and $k=t+1$ via a \emph{restricted sum} construction.

Let ${\mathbb Z}_v$ denote the integers modulo $v$. First, suppose that $v=2\ell$ is even. Consider $k$-subsets of ${\mathbb Z}_v$ that sum to $0\mod v$, where $k$ is fixed. Note that when any $k-1$ integers in ${\mathbb Z}_v$ are chosen, the last integer is uniquely determined.  

When $k = 4m$, choose any $2m-1$ different even integers from ${\mathbb Z}_v$, and choose $2m$ different odd integers from ${\mathbb Z}_v$. The integer that makes the sum equal to $0$ must be even; hence, either the set of $4m$ elements contains the same number of even and odd integers, or the last element repeats an element already chosen. In the latter case, we exclude the set from the choice of $4m-1$ elements from further consideration.  

When $k = 4m+1$, choose any $2m$ different even integers from ${\mathbb Z}_v$, and choose $2m$ different odd integers from ${\mathbb Z}_v$. The integer that makes the sum equal to $0$ must be even; hence, either the set of $4m+1$ elements contains one more even than odd integer, or that the last element repeats an element already chosen. In the latter case, the set from the choice of $4m$ elements is excluded.

When $k = 4m+2$, choose any $2m+1$ different even integers from ${\mathbb Z}_v$, and choose $2m$ different odd integers from ${\mathbb Z}_v$. The integer that makes the sum equal to $0$ must be odd; hence we must have the same number of even and odd, or that the last element repeats an element already chosen. As before, in the latter case, the set from the choice of $4m$ elements is excluded.

When $k = 4m+3$, choose any $2m+1$ different even integers from ${\mathbb Z}_v$, and choose $2m+1$ different odd integers from ${\mathbb Z}_v$. The integer that makes the sum equal to $0$ must be odd; hence, our selection, if not including repeated elements, has one less even than odd number.

By labeling every even integer by a $+1$ and odd integer by $-1$, we obtain desired discrepancy values for each block. The number of blocks obtained via this procedure may be computed as follows.
As an example, consider the case $k=4m+3$; we have $\binom{\ell}{2m+1}\binom{\ell}{2m+1}$ sets in our initial selection, and these selections were completed in any of $2m+2$ ways by choosing different subsets of the odd elements. This leads to $\binom{\ell}{2m+1}\binom{\ell}{2m+1}/(2m+2)$ different sets, provided that the last element added is not a repetition. Hence, we only need to determine how many times the forced last element causes a repetition and our selection fails. Roughly, there are $\ell$ odd numbers, $2m+1$ selected odd numbers that may lead to a repetition, so that the fraction of the times that we fail is $\frac{2m+1}{\ell}$. Since $\ell=\frac{v}{2}$, whenever $k$ is fixed and $v \rightarrow \infty$, this failure rate approaches zero. A similar analysis can be performed for other selections of $k$. Hence, asymptotically, the sum construction is extremal.

As an example, consider the case $t=2, \, k=3$ and $v=12$. In this case, we get the following $15$ blocks:
\begin{align*}
&\{ 0,1,11 \}, \{ 0,3,9\}, \{ 0,5,7\}, \{ 1,2,9 \},\\
&\{ 1,4,7 \}, \{ 1,5,6 \}, \{ 1,3,8 \}, \{ 2,3,7 \}, \\
&\{ 3,4,5 \}, \{ 3,10,11\}, \{ 4,9,11\}, \{5,8,11\},\\
&\{ 5,9,10 \}, \{ 6,7,11\}, \{ 7,8,9\}.
\end{align*} 
We miss all pairs when both initial choices are even, and in general, we miss all initial pairs of the form $\{x,-2x\}$ as $x-2x+x = 0$. Furthermore, we miss the triples obtained from 
$$\{ 1,10\}, \{2,5\}, \{2,11\}, \{3,6\}, \{6,9\}, \{7,10\}.$$ 
These are the only triples omitted, because if in $\{x,y,z\}$ we have $x = z$, then we must have $y = -2x$. By comparison, the extremal construction from the previous section produces $\frac{v^2}{8}=18$ triples.

These observations easily extend to any choice of even $v$ and any $k=t+1$, but they do not yield simple counterparts for the case of odd $v$.

\subsection{Additional choices of design parameters with optimal constructions}
\subsection{Constructions based on transversal designs}

Next, we return to the approach outlined at the beginning of our discussion, in which balanced families of set with small intersections are obtained using transversal designs (akin to the Babai-Frankl method). The next lemma addresses the case $k=4$ and $t=3$ and provides extremal balanced sets. It also rigorously specifies how to identify additional blocks to the design so that the result is optimal.  

\begin{lemma}\label{lem:t34} Suppose that there exists a $TD(t,k,v)$ with $v=4m$ and $m$ even. Then,
$$A(3,4,v) = \frac{\nchoosek{v/2}{2}  \nchoosek{v/2}{1}}{2}=\frac{v^2(v/2-1)}{16}.$$   
\end{lemma}
\begin{proof} For simplicity, we assume that points in our design are represented as ordered pairs in $V = \mathbb{Z}_m \times [4],$ where the labels $\{{0,1,2,3\}}$ refer to the groups in the transversal design. According to the labeling approach from the introduction, we set $P_{+} = [m] \times \{0,1\}$ and $P_{-} = [m] \times \{2,3\}$. There are $m^3$ blocks in the transversal design, all of which are balanced based on our choice of labels and the the inherent properties of the transversal design. 

We now describe formally how to add new blocks to the blocks in the design to obtain extremal families. The blocks are of the form 
$$ \{a_1, a_2, b_{1}, b_2 \},$$ 
where $a_1, a_2 \in [m] \times \{0,1\}$, $b_{1}, b_2 \in [m] \times \{2,3\}$. Let $\Phi^{(g)} = \{\Phi^{(g)}_1, \ldots, \Phi^{(g)}_{m-1} \}$ be a $1$-factorization of points in group $g$ of the transversal design, where $g=0,1,2,3$. Select $\{ a_1, a_2 \} \in \Phi^{(g_1)}_i$ and  $\{ b_{1}, b_2 \} \in \Phi^{(g_2)}_i$, where $g_1 \in \{{0,1\}}$ and $g_2 \in \{{2,3\}}$, and $1 \leq i \leq m-1$.

It is straightforward to see that there are $2m \times \nchoosek{m}{2}$ blocks of the above form. There are four ways to select the pair $g_1 \in \{0,1\}$ and $g_2 \in \{2,3\}$; for $g_1$ and $g_2$ fixed, there are $\nchoosek{m}{2}$ choices for $a_1, a_2$. Since $\Phi^{(g_1)}$ is a $1$-factorization, there exists some $1 \leq i \leq m$, such that $\{a_1, a_2\} \in \Phi^{(g_1)}_i$. Next, we select a pair $\{b_{1}, b_2\}$ from the set $\Phi^{(g_2)}_i$. Since $\Phi^{(g_2)}_f$ contains $\frac{m}{2}$ disjoint edges, there are $\frac{m}{2}$ choices for $\{b_{1}, b_2\}$. Thus, in total, there are
$$ 4 \times \nchoosek{m}{2} \times \frac{m}{2} = 2m \times \nchoosek{m}{2},$$
choices for the new blocks. It is straightforward to see that no two newly added blocks intersect in more than two points and that the blocks are balanced.

The union of the blocks from the transversal design and the newly constructed blocks has 
$$m^3 + 2m \times \nchoosek{m}{2}  = \frac{\nchoosek{v/2}{2}  \nchoosek{v/2}{1}}{2}$$ 
elements, which proves the claimed bound.
\end{proof}

\begin{example} Suppose that $S_{TD}(3,4,16)$ is a transversal design whose with points of the form $\mathbb{Z}_4 \times \{0,1,2,3\}$, and blocks 
\begin{align*}
\Big \{ \{ (i_0,0), (i_1,1), (i_2,2), (i_3,3) \} : i_0,i_1,i_2,i_3 \in \mathbb{Z}_4, \, i_3 = i_0 + i_1 + i_2 \Big \}.
\end{align*}
Based on the outlined construction, we set $P_{+} = \mathbb{Z}_4 \times \{0,1\}$ and $P_2 = \mathbb{Z}_4 \times \{2,3\}$. Since there are four choices for each of $i_0, i_1, i_{2}$ it follows that we have $4^3 = 64$ blocks.

The factorization of interest reads as 
\begin{align*}
\Phi^{(g)} = \{ \Phi^{(g)}_1, \Phi^{(g)}_2, \Phi^{(g)}_3 \} = \Bigg \{ &\Big\{ \{(0,g),(1,g)\}, \{(2,g),(3,g)\} \Big\}, \\
&\Big \{ \{(0,g), (2,g) \}, \{ (1,g), (3,g) \} \Big \}, \\
&\Big \{ \{(0,g), (3,g) \}, \{ (1,g), (2,g) \} \Big \}\Bigg \},
\end{align*}
where $g \in \{0,1,2,3\}$. Thus, for a $g_1 \in \{0,1\}$ and a $g_2 \in \{2,3\}$, we obtain $12$ blocks:
\begin{align*}
\Bigg \{  &\Big \{ (0,g_1), (1,g_1),(0,g_2), (1,g_2) \Big \}, \Big \{ (0,g_1), (1,g_1), (2,g_2), (3,g_2)  \Big \}, \Big \{ (2,g_1), (3,g_1), (0,g_2), (1,g_2)  \Big \},\\
&\Big \{ (2,g_1), (3,g_1), (2,g_2), (3,g_2) \Big \}, \Big \{ (0,g_1), (2,g_1), (0,g_2), (2,g_2) \Big \}, \Big \{ (0,g_1), (2,g_1), (1,g_2), (3,g_2) \Big \}, \\
&\Big \{ (1,g_1), (3,g_1), (0,g_2), (2,g_2) \Big \}, \Big \{ (1,g_1), (3,g_1), (1,g_2), (3,g_2) \Big \}, \Big \{ (0,g_1), (3,g_1), (0,g_2), (3,g_2) \Big \}, \\
&\Big \{ (0,g_1), (3,g_1), (1,g_2), (2,g_2) \Big \}, \Big \{ (1,g_1), (2,g_1), (0,g_2), (3,g_2) \Big \}, \Big \{ (1,g_1), (2,g_1), (1,g_2), (2,g_2) \Big \} \Bigg \}.
\end{align*}
Since there are four choices for $g_1, g_2$, it follows that the total number of added blocks is $48$. Thus, we have a balanced $(3,4,16)$ packing with $112$ blocks, which is optimal.
\end{example}

When $v$ is a power of two, the construction of Lemma~\ref{lem:t34} can be simplified further. 

\begin{example} The points are of the form $\mathbb{F}_{2^m} \times \{0,1\}$, and $P_{+} = \{ (x,0) : x \in \mathbb{F}_{2^m} \}$ and $P_{-} = \{ (x,1) : x \in \mathbb{F}_{2^m} \}$. The additional blocks are generated according to
$$\Big \{ \{ (a_1,0), (a_2,0), (b_{1},1), (b_2,1) \} : a_1, a_2, b_{1}, b_2 \in \mathbb{F}_{2^m}, b_2 = a_1 + a_2 + b_{1} \Big \}.$$
In order to verify the validity of the construction, one needs to prove that for any choice of distinct $a_1,a_1$ and $b_{1} \neq b_2$, there exists a $b_2$ that satisfies the defining constraint. Since $a_1, a_2, b_{1}, b_2$ belong to a field with characteristic two, if $b_{1}= b_2$, it follows that $a_1 = a_2$, which is impossible. Therefore, the blocks in the transversal design along with the blocks from the above construction result in a $(t=3,k=4,v)$ extremal balanced system whenever $v$ is a power of two.
\end{example}

\subsection{Constructions based on maximal disjoint Steiner systems}

When the sets in ${\mathcal F}$ can be partitioned into classes $\mathcal{F}_1, \dots, \mathcal{F}_n$ so that $(V,\mathcal{F}_i)$ is a $(t',k,v)$ packing for each $1 \leq i \leq n$, we say that $(V,{\mathcal F})$ is $t'$-{\em partitionable} with partition classes $\mathcal{F}_1, \dots, \mathcal{F}_n$. Let $(V_1,{\mathcal F}^1)$ be a $(t_1,k_1,v_1)$ packing that is $t_1$-partitionable with partition classes $\mathcal{F}^1_1, \dots, \mathcal{F}^1_n$. Similarly, let $(V_2,{\mathcal F}^2)$ be a $(t_2,k_2,v_2)$ packing that is $t_2$-partitionable with partition classes $\mathcal{F}^2_1, \dots, \mathcal{F}^2_n$. Suppose that $V_1$ and $V_2$ are disjoint.  
Form a new packing with blocks $\{ B \cup D : B \in \mathcal{F}^1_i, D \in \mathcal{F}^2_i, 1 \leq i \leq n\}$.
This packing has $v= v_1 + v_2$ points.  Each block has $k = k_1 + k_2$ points.  
Two distinct blocks can share at most $\max(k_2+t_1-1,k_1+ t_2-1)$ points. Hence, the resulting structure is a $(\max(k_2+t_1,k_1+ t_2), k_1 + k_2, v_1 + v_2)$ packing. By setting $V_1=P_+$ and $V_2=P_{-}$, each block has discrepancy $|k_1-k_2|$. Therefore, we need to choose $k_1$ and $k_2$ to have values as close as possible. 

\subsubsection{Balanced sets with parameters $(2,3,v)$}

A $1$-factorization is a $(2,2,2m)$ packing that is $1$-partitionable into $2m-1$ classes.  
A set of $2m-1$ points, each forming a block of size one, is a $(1,1,2m-1)$ packing that is $0$-partitionable into $2m-1$ classes. Hence, using this factorization approach we can obtain a $(2,3,4m-1)$ packing with discrepancy $1$ having $m(2m-1)$ blocks, or roughly, $\frac{1}{8}\, v^2$ blocks. In comparison, a Steiner triple system would have $(4m-1)(4m-2)/6$ blocks, which roughly equals $\frac{8}{3}m^2$ ($\frac{1}{6}\, v^2$). 

\subsubsection{Balanced sets with parameters $(3,4,v)$}

A $1$-factorization is a $(2,2,2m)$ packing that is $1$-partitionable into $2m-1$ classes. Hence, we have a $(3,4,4m)$ packing with discrepancy $0$ and $m^2(2m-1)$ blocks, or roughly, $\frac{1}{32}\, v^3$ blocks. In comparison, a Steiner quadruple system would have $(4m)(4m-1)(4m-2)/24$ blocks, which roughly equals $\frac{8}{3}m^3$ ($\frac{1}{24}\, v^3$).

\subsubsection{Balanced sets with parameters $(4,5,v)$}

A $1$-factorization is a $(2,2,2m)$ packing that is $1$-partitionable into $2m-1$ classes.  
A \emph{large set of Steiner triple system or maximal disjoint Steiner system}~\cite{lu1983large,teirlinck1975maximum} is a $(3,3,v)$ packing that is $2$-partitionable into $v-2$ classes, or equivalently, a set of $v-2$ Steiner triple systems that have disjoint block sets. Such a systems exists whenever $(2m+1 =) \, v \equiv 1,3 \pmod{6}$ (with six exceptions, see~\cite{lu1983large,teirlinck1975maximum}).
So we obtain a $(4,5,4m+1)$ packing with discrepancy $1$ having $(2m-1)m^2(2m+1)/6$ blocks, which is roughly  $\frac{1}{384}v^4$ blocks. A Steiner quintuple system would have $(4m+1)(4m)(4m-1)(4m-2)/120$ blocks, which roughly equals $\frac{32}{15}m^4$ ($\frac{1}{120}v^4$).

\subsubsection{Balanced sets with parameters $(5,6,v)$}

A large set of Steiner triple systems is a $(3,3,v)$ system that is 2-partitionable into $v-2$ classes, when $(2m+1 =) \, v \equiv 1,3 \pmod{6}$ (with six exceptions, see~\cite{lu1983large,teirlinck1975maximum}). So we obtain a $(5,6,4m+2)$ packing with discrepancy $0$ having roughly $(2m-1)\left ( \frac{(2m+1)m}{3} \right ) ^2 \approx \frac{8}{3}m^5$ blocks. A Steiner sixtuple system would have $(4m+1)(4m)(4m-1)(4m-2)/120$ blocks, which roughly equals $\frac{32}{15}m^4$ ($\frac{1}{120}v^4$).

\subsubsection{Balanced sets for other parameter choices}

Next, we consider two construction for general $t$ and $k$. For this construction, once again we make use of the notion of maximal disjoint Steiner system~\cite{kramer1974intersections}. 
We denote the decomposition of the maximal disjoint Steiner system by $\mathcal{M}(t,k,v)=\{ \mathcal{M}_1(t,k,v), \ldots, \mathcal{M}_{T}(t,k,v)\}$ where for any $i \in [v-T]$, $\mathcal{M}_i(t,k,v)$ is a $(t,k,v)$ Steiner system, and for any two distinct $i, j \in [T]$, $\mathcal{M}_{i}(t,k,v) \cap \mathcal{M}_{j}(t,k,v) = \emptyset$. Here,
$$T = \frac{\nchoosek{v}{k}}{\nchoosek{v}{t} \Big / \nchoosek{k}{t}}.$$ 
We first consider the case when $k$ is even.

\begin{construction}\label{con:mds} Suppose that there exists a maximal disjoint Steiner system $\mathcal{M}(t,k,v)$. Let $P_{+} = \{ (1,0), (2,0), \ldots, (v,0) \}$ and $P_{-} = \{ (1,1), (2,1), \ldots, (v,1) \}$ and form the following sets
\begin{align*}
\S_{D}(t,t+1,2v) = \Big \{ &\big \{ (a_1, 0), \ldots, (a_k,0),(B_{+}, 1), \ldots, (b_k,1) \big \} : \\
&\{a_1, \ldots, a_k \}, \{B_{+}, \ldots, b_k\} \in \mathcal{M}_i(t,k,v), i \in  T \Big \}.
\end{align*}
\end{construction}

The following theorem is a straightforward result that is based on Construction~\ref{con:mds}.

\begin{theorem}\label{th:mds} Suppose that there exists a maximal disjoint Steiner system $\mathcal{M}(t,k,v)$. Then, there exists an extremal $(t+k,2k,2v)$ balanced packing. \end{theorem}

A drawback of Construction~\ref{con:mds} is that relatively few constructions for maximal disjoint Steiner systems exist, with most results focusing on Steiner triple systems. It was shown in~\cite{lu1983large} that if the order of $2$ modulo $v-2$ is an odd number then there exists a maximal disjoint Steiner triple system. In addition, it is known that if there exists a maximal disjoint Steiner triple system with point set $[v]$, then for $\pi$ a prime number there also exists a maximal disjoint Steiner triple system over $[\pi v + 2]$. An immediate consequence of Theorem~\ref{th:mds} is the following result.

\begin{lemma} Suppose there exists a maximal disjoint system $\mathcal{M}(2,3,v)$. Then, there exists a $(5,6,2v)$ balanced packing of maximum size.  \end{lemma}

Using a simple modification of Construction~\ref{con:mds}, we arrive at the following result.

\begin{lemma} Suppose that there exists a maximal disjoint system $\mathcal{M}(2,3,p_{+})$ and that $p_{-} = p_{+} - 1$ is an even number. Then, there exists a $(4,5,p_{+}+p_{-})$ balanced packing of maximum size.
\end{lemma}
\begin{proof} There are $T = p_{+}-2$ disjoint Steiner triple systems 
$\{ \mathcal{M}_1(2,3,p_{+}), \ldots,\mathcal{M}_{p_{+}-2}(2,3,p_{+}) \},$ where $p_{+}-2=p_{-}-1$. If $p_{-}$ is even, then there exists a $1$-factorization of $\Phi=\{ \Phi_1, \ldots, \Phi_{p_{+}-2} \}$ of $K_{p_{-}}$. Let 
\begin{align*}
\Big \{ &\big \{ (a_1, 0), (a_2,0), (a_3,0),(b_{1}, 1), (b_2,1) \big \} : \\
&\{a_1, a_2, a_3 \} \in \mathcal{M}_i(2,3,p_{+}), \{b_{1}, b_2\} \in \Phi_i, i \in [T] \Big \}.
\end{align*}
The proof that the described set of blocks constitutes a $(4,5,p_{+}+p_{-})$ extrema balanced packing follows from the same line of reasoning as used in Theorem~\ref{th:mds}.
\end{proof}

\section{Acknowledgement}
The authors gratefully acknowledge discussions with Joao Ribeiro, Imperial College, London. 
The work was supported in part by the NSF grants CCF 1526875, 1816913, and 1813729. 

\bibliographystyle{plain}
\bibliography{intersectingsets}

\begin{thebibliography}{10}

\bibitem{armoni1996discrepancy}
Roy Armoni, Michael Saks, Avi Wigderson, and Shiyu Zhou.
\newblock Discrepancy sets and pseudorandom generators for combinatorial
  rectangles.
\newblock In {\em Foundations of Computer Science, 1996. Proceedings., 37th
  Annual Symposium on}, pages 412--421. IEEE, 1996.

\bibitem{BabaiFrankl2004}
L.~Babai and P.~Frankl.
\newblock {\em Linear Algebra Methods in Combinatorics with Applications to
  Geometry and Computer Science}.
\newblock Department of Computer Science, University of Chicago, New York,
  1992.

\bibitem{beck1981balanced}
J{\'o}zsef Beck.
\newblock Balanced two-colorings of finite sets in the square i.
\newblock {\em Combinatorica}, 1(4):327--335, 1981.

\bibitem{Bose1939}
R.~C. Bose.
\newblock On the construction of balanced incomplete block designs.
\newblock {\em Annals of Eugenics}, 9:353--399, 1939.

\bibitem{bose1952orthogonal}
Raj~C Bose and K~Ak Bush.
\newblock Orthogonal arrays of strength two and three.
\newblock {\em The Annals of Mathematical Statistics}, pages 508--524, 1952.

\bibitem{bujtas2015transversal}
Csilla Bujt{\'a}s and Zsolt Tuza.
\newblock Transversal designs and induced decompositions of graphs.
\newblock {\em arXiv preprint arXiv:1501.03518}, 2015.

\bibitem{colbourn2010crc}
Charles~J Colbourn and Jeffrey~H Dinitiz.
\newblock {\em CRC handbook of combinatorial designs}.
\newblock CRC press, 2010.

\bibitem{colbourn1999bicoloring}
Charles~J Colbourn, Jeffrey~H Dinitz, and Alexander Rosa.
\newblock Bicoloring steiner triple systems.
\newblock {\em the electronic journal of combinatorics}, 6(1):25, 1999.

\bibitem{doerr2001lattice}
Benjamin Doerr.
\newblock Lattice approximation and linear discrepancy of totally unimodular
  matrices--extended abstract.
\newblock In {\em SIAM-ACM Symposium on Discrete Algorithms}, pages 119--125,
  2001.

\bibitem{doerr2003multicolour}
Benjamin Doerr and Anand Srivastav.
\newblock Multicolour discrepancies.
\newblock {\em Combinatorics, Probability and Computing}, 12(4):365--399, 2003.

\bibitem{erdos1946asymptotic}
Paul Erd{\"o}s and Irving Kaplansky.
\newblock The asymptotic number of latin rectangles.
\newblock {\em American Journal of Mathematics}, 68(2):230--236, 1946.

\bibitem{gabrys2017asymmetric}
Ryan Gabrys, Han~Mao Kiah, and Olgica Milenkovic.
\newblock Asymmetric lee distance codes for dna-based storage.
\newblock {\em IEEE Transactions on Information Theory}, 63(8):4982--4995,
  2017.

\bibitem{hedayat2012orthogonal}
A~Samad Hedayat, Neil James~Alexander Sloane, and John Stufken.
\newblock {\em Orthogonal arrays: theory and applications}.
\newblock Springer Science \& Business Media, 2012.

\bibitem{kiah2016codes}
Han~Mao Kiah, Gregory~J Puleo, and Olgica Milenkovic.
\newblock Codes for dna sequence profiles.
\newblock {\em IEEE Transactions on Information Theory}, 62(6):3125--3146,
  2016.

\bibitem{kramer1974intersections}
Earl~S Kramer and Dale~M Mesner.
\newblock Intersections among steiner systems.
\newblock {\em Journal of Combinatorial Theory, Series A}, 16(3):273--285,
  1974.

\bibitem{lovasz1986discrepancy}
L{\'a}szl{\'o} Lov{\'a}sz, Joel Spencer, and Katalin Vesztergombi.
\newblock Discrepancy of set-systems and matrices.
\newblock {\em European Journal of Combinatorics}, 7(2):151--160, 1986.

\bibitem{lu1983large}
Jia-Xi Lu.
\newblock On large sets of disjoint steiner triple systems i,ii, iii, iv.
\newblock {\em J. Comb. Theory, Ser. A}, 34(2):140--146, 1983.

\bibitem{matouvsek1993discrepancy}
Ji{\v{r}}{\'\i} Matou{\v{s}}ek, Emo Welzl, and Lorenz Wernisch.
\newblock Discrepancy and approximations for bounded vc-dimension.
\newblock {\em Combinatorica}, 13(4):455--466, 1993.

\bibitem{milazzo2003strict}
Lorenzo Milazzo, Zsolt Tuza, and Vitaly Voloshin.
\newblock Strict colorings of steiner triple and quadruple systems: a survey.
\newblock {\em Discrete Mathematics}, 261(1-3):399--411, 2003.

\bibitem{milenkovic2018exabytes}
Olgica Milenkovic, Ryan Gabrys, Han~Mao Kiah, and SM~Hossein~Tabatabaei Yazdi.
\newblock Exabytes in a test tube.
\newblock {\em IEEE Spectrum}, 55(5):40--45, 2018.

\bibitem{muthukrishnan2012optimal}
S~Muthukrishnan and Aleksandar Nikolov.
\newblock Optimal private halfspace counting via discrepancy.
\newblock In {\em Proceedings of the forty-fourth annual ACM symposium on
  Theory of computing}, pages 1285--1292. ACM, 2012.

\bibitem{nicking}
H.~Zhao A.~Hernandez O.~Milenkovic, K.~Tabatabei.
\newblock Nick-based storage in native nucleic acids.
\newblock {\em US Patent, UIUC2017-170-02(US) // MBHB 18-1205}, 2018.

\bibitem{rothvoss2013approximating}
Thomas Rothvo{\ss}.
\newblock Approximating bin packing within o (log opt* log log opt) bins.
\newblock In {\em Foundations of Computer Science (FOCS), 2013 IEEE 54th Annual
  Symposium on}, pages 20--29. IEEE, 2013.

\bibitem{saks2000low}
Michael Saks, Aravind Srinivasan, Shiyu Zhou, and David Zuckerman.
\newblock Low discrepancy sets yield approximate min-wise independent
  permutation families.
\newblock {\em Information Processing Letters}, 73(1-2):29--32, 2000.

\bibitem{Skolem1958}
T.~Skolem.
\newblock Some remarks on the triple systems of {S}teiner.
\newblock {\em Mathematica Scandinavica}, 6:273--280, 1958.

\bibitem{ygm17}
O.~Milenkovic S.M.T.~Yazdi, R.~Gabrys.
\newblock Portable and error-free dna-based data storage.
\newblock {\em Scientific Reports}, 7(5011), 2017.

\bibitem{solymosi2009incidences}
Jozsef Solymosi.
\newblock Incidences and the spectra of graphs.
\newblock In {\em Combinatorial Number Theory and Additive Group Theory}, pages
  299--314. Springer, 2009.

\bibitem{Stinson2004}
D.~R. Stinson.
\newblock {\em Combinatorial Designs: Constructions and Analysis}.
\newblock Springer, New York, 2004.

\bibitem{stinson1981general}
Douglas~R Stinson.
\newblock A general construction for group-divisible designs.
\newblock {\em Discrete Mathematics}, 33(1):89--94, 1981.

\bibitem{teirlinck1975maximum}
Luc Teirlinck.
\newblock On the maximum number of disjoint triple systems.
\newblock {\em Journal of Geometry}, 6(1):93--96, 1975.

\bibitem{trevisan2001extractors}
Luca Trevisan.
\newblock Extractors and pseudorandom generators.
\newblock {\em Journal of the ACM}, 48(4):860--879, 2001.

\bibitem{yazdi2017portable}
SM~Hossein~Tabatabaei Yazdi, Ryan Gabrys, and Olgica Milenkovic.
\newblock Portable and error-free dna-based data storage.
\newblock {\em Scientific reports}, 7(1):5011, 2017.

\bibitem{yazdi2018mutually}
SM~Hossein~Tabatabaei Yazdi, Han~Mao Kiah, Ryan Gabrys, and Olgica Milenkovic.
\newblock Mutually uncorrelated primers for dna-based data storage.
\newblock {\em IEEE Transactions on Information Theory}, 2018.

\bibitem{yazdi2015rewritable}
SM~Hossein~Tabatabaei Yazdi, Yongbo Yuan, Jian Ma, Huimin Zhao, and Olgica
  Milenkovic.
\newblock A rewritable, random-access dna-based storage system.
\newblock {\em Scientific reports}, 5:14138, 2015.

\end{thebibliography}

\appendices

\section{Proof of Lemma~\ref{lem:odd}}\label{app:lemodd}

Suppose that, to the contrary, there exists a $(t,k,v)$ balanced system $S(t,k,v)$, which according to repeated application of Claim~\ref{cl:recurse} means that there exists a balanced system $S(3,k-t+3,v-t+3)$. Next, we proceed similarly as before, except that we select one more element $e_1 \in P_{-}$ and let 
 $$ S(2,k-t+2,v-t+2) = \Big \{ B \setminus e_1 : B \in S(3,k-t+3,v-t+3) \text{ and } e_1 \in B \Big \}. $$
Let $\hat{P}_{+}$ denote the elements in $S(2,k-t+2,v-t+2)$ labeled $+1$ and let $\hat{P}_{-}$ denote the elements labeled $-1$. Again, let $\hat{p}_+ = |\hat{P}_{+}|, \hat{p}_- = |\hat{P}_{-}|$, and $\hat{k} = \lceil \frac{k-t+2}{2} \rceil$. Note that under this setup, for any $B \in S(2,k-t+2, v-t+2)$, the discrepancy of $B$ is either $-2, -1,$ or $0$. More specifically, if $k-t+2$ is odd, then $|B \cap \hat{P}_{+}|=\hat{k}$. If $k-t+2$ is even, then $|B \cap \hat{P}_{+}| \in \{ \hat{k}, \hat{k}-1 \}$.

The next claim is an analogue of Claim~\ref{cl:k=p}.

\begin{claim}\label{cl:k=p2} One has $\hat{p}_+ > \hat{k}$. \end{claim}
\begin{proof} If $S(3,k-t+3,v-t+3)$ is obtained from $S(t,k,v)$ by applying the procedure from Claim~\ref{cl:recurse}, then $\hat{p}_{1} = p_{+} - \frac{t-3}{2}$. Thus, if $p_{+} > \frac{k+1}{2}$, which is one of the assumptions in the lemma, then $\hat{p}_+ > \hat{k} = \lceil \frac{k-t+2}{2} \rceil$.
\end{proof}

Next, suppose first that $k-t+2$ is odd so that for any $B \in S(2,k-t+2,v-t+2)$, $|B \cap \hat{P}_{+}| = \hat{k}$ and $|B \cap \hat{P}_{-}| = \hat{k}-1$. Then, for a Steiner system $S(2,k-t+2,v-t+2)$, counting arguments similar to those used in the previous lemma would lead to
$$\frac{ \nchoosek{\hat{p}_+}{2}}{\nchoosek{\hat{k}}{2}} = \frac{\hat{p}_+  \hat{p}_-}{\hat{k}  (\hat{k}-1)},$$
which means that $\hat{p}_+-1 = \hat{p}_-.$ Note that we can count the number of elements in $S(2,k-t+2,v-t+2)$ by also considering all pairs from $\hat{P}_{-}$, which implies 
$$ \frac{ \nchoosek{\hat{p}_-}{2}}{\nchoosek{\hat{k-1}}{2}} = \frac{ \nchoosek{\hat{p}_+}{2}}{\nchoosek{\hat{k}}{2}}.$$
If $\hat{p}_+-1 = \hat{p}_-$, then this implies that $\hat{p}_+ = \hat{k}$ which contradicts Claim~\ref{cl:k=p2}.

Next, we consider the case where $k-t+2$ is even. Let $B_{+} \subset S(2,k-t+2,v-t+2)$ denote the set of blocks in $S(2,k-t+2,v-t+2)$ with $\hat{k}+1$ elements in $\hat{P}_{+}$ and that have $\hat{k}-1$ elements from $\hat{P}_{-}$. Similarly, let $B_{-} \subset S(2,k-t+2,v-t+2)$ denote the set of blocks in $S(2,k-t+2,v-t+2)$ with $\hat{k}$ elements in $\hat{P}_{+}$ and $\hat{k}$ elements from $\hat{P}_2$. Using similar ideas as the previous lemma, we arrive at
\begin{align}\label{eq:evB1}
|B_{+}| \nchoosek{\hat{k}+1}{ 2} + |B_{-}| \nchoosek{\hat{k}}{2} = \nchoosek{\hat{p}_+}{2},
\end{align}
and 
\begin{align}\label{eq:evB2}
|B_{+}| \nchoosek{\hat{k}-1}{ 2} + |B_{-}| \nchoosek{\hat{k}}{2} = \nchoosek{\hat{p}_-}{2}.
\end{align}
From (\ref{eq:evB1}) and (\ref{eq:evB2}), we have that
\begin{align*}
|B_{+}| = \frac{\nchoosek{\hat{p}_+}{2} - \nchoosek{\hat{p}_-}{2}}{\nchoosek{\hat{k}+1}{2} - \nchoosek{\hat{k}-1}{2}},
\end{align*}
and
\begin{align*}
|B_{-}| = \frac{\nchoosek{\hat{p}_+}{2}  - |B_{+}| \nchoosek{\hat{k}+1}{2}}{\nchoosek{\hat{k}}{2}}.
\end{align*}

We now follow the same logic as in the previous proof. First, notice that according to (\ref{eq:evB1}) and (\ref{eq:evB2}), the total number of blocks in $S(2,k-t+2,v-t+2)$ is

$$|B_{+}| + |B_{-}| = \frac{\hat{p}_+^2(\hat{k}-1) + \hat{p}_2(\hat{p}_2 - 1)\hat{k} + \hat{p}_+ - \hat{p}_+ \hat{k}}{(2\hat{k}-1)\hat{k}(\hat{k}-1)}.$$
We also know that since $S(2,k-t+2,v-t+2)$ is a Steiner system, we the number of blocks in $S(2,k-t+2,v-t+2)$ is $\frac{(\hat{p}_++\hat{p}_2)(\hat{p}_++\hat{p}_2-1)}{2\hat{k}(2\hat{k}-1)}$, which implies
$$ \frac{\hat{p}_+^2(\hat{k}-1) + \hat{p}_2(\hat{p}_2 - 1)\hat{k} + \hat{p}_+ - \hat{p}_+ \hat{k}}{(2\hat{k}-1)\hat{k}(\hat{k}-1)} = \frac{(\hat{p}_++\hat{p}_2)(\hat{p}_++\hat{p}_2-1)}{2\hat{k}(2\hat{k}-1)}.$$
Solving for $\hat{k}$ gives
$$ \hat{k} = \frac{2\hat{p}_+\hat{p}_- + \hat{p}_-^2 - \hat{p}_+^2 + \hat{p}_+ - \hat{p}_-}{2\hat{p}_+\hat{p}_- - \hat{p}_-^2 - \hat{p}_+^2 + \hat{p}_+ + \hat{p}_-}.$$

Recall from Claim~\ref{cl:cases} that we only have two cases left to consider: The case where $\hat{p}_+ = \hat{p}_-$ and where $\hat{p}_+ = \hat{p}_- + 1$. If $\hat{p}_+ = \hat{p}_-$, then $\hat{k} = \hat{p}_+$ which is a contradiction of Claim~\ref{cl:k=p2}. Otherwise, if $\hat{p}_+ = \hat{p}_- + 1$, then $\hat{k} = \hat{p}_+-1$, and we arrive at a contradiction in this case as well.

\end{document}